\documentclass[runningheads]{llncs}

\usepackage[latin2]{inputenc}
\usepackage[english]{babel}
\usepackage{amsmath}
\usepackage{amssymb}
\usepackage{graphicx}
\usepackage{xcolor}
\usepackage{tikz}
\usepackage{multirow}
\newcommand{\mycircled}[2][none]{%
	\tikz[baseline=(a.base)]\node[draw,circle,inner sep=1pt, outer sep=0pt,fill=#1](a){\ensuremath #2\strut};
}
\newcommand{\mysquared}[2][none]{%
	\tikz[baseline=(a.base)]\node[draw,rectangle,inner sep=2pt,  outer sep=0pt,fill=#1](a){\ensuremath #2\strut};
}

\usepackage{caption}

\newcommand{\F}{\mathbb{F}}

\newcommand{\T}{{\rm Tr}}

\begin{document}
		\title{On functions of low differential uniformity\\ in characteristic 2: A close look\\ (I)} 
		
		%\thanks{Nurdag\"{u}l Anbar and Tekg\"{u}l Kalayc\i \  are supported by T\"UB\.ITAK Project under Grant 120F309 }}
	%
	\titlerunning{On functions of low differential uniformity}
	% If the paper title is too long for the running head, you can set
	% an abbreviated paper title here
	%
	\author{ Nurdag\"{u}l Anbar\and
Tekg\"{u}l Kalayc\i	\and
		Alev Topuzo\u{g}lu }
	\authorrunning{N. Anbar et al.}
	% First names are abbreviated in the running head.
	% If there are more than two authors, 'et al.' is used.
	%
	\institute{Sabanc{\i} University,\\
		 MDBF, Orhanl\i, Tuzla, 34956 \. Istanbul, Turkey\\
		\email{nurdagulanbar2@gmail.com}\\
		\email{tekgulkalayci@sabanciuniv.edu}
		\email{alev@sabanciuniv.edu}}

\maketitle

\begin{abstract} 
We introduce a new concept, the\textit{ APN-defect}, which can be thought of as measuring the distance of a given function $G:\mathbb{F}_{2^n} \rightarrow \mathbb{F}_{2^n}$ to the set of almost perfect nonlinear (APN) functions. 
This concept is motivated by the detailed analysis of the differential behaviour of non-APN functions (of low differential uniformity) $G$ using the so-called \textit{difference squares}. 
We describe the relations between the APN-defect and other recent concepts of similar nature. Upper and lower bounds for the values of APN-defect for several classes of functions of interest, including Dembowski-Ostrom polynomials are given. Its exact values in some cases are also calculated. 
The difference square corresponding to a modification of the inverse function is determined, its APN-defect depending on $n$ is evaluated and the implications are discussed.\\
In the forthcoming second part of this work we further examine modifications of the inverse function. We also study modifications of  classes of functions of low uniformity
over infinitely many extensions of $\F_{2^n}$. We present quantitative results on their differential behaviour, especially in connection with their APN-defects.


\end{abstract}

\section{Introduction }

Let $\F_q$ be the finite field with $q$ elements, where $q=2^n$, for some positive integer $n$.

The so-called substitution boxes (S-boxes) in block ciphers are multiple output  (or vectorial) Boolean functions $f:\mathbb{F}_2^{m} \rightarrow \mathbb{F}_2^{n}$. They are also called $(m,n)$ 
functions. There is extensive work on the construction of $(m,n)$ 
functions in such a way that they satisfy a good number of relevant criteria in order to be resistant against most known attacks. For a detailed exposition of $(m,n)$ functions and such criteria, we refer the reader to \cite[Sections 1.4, 3.2-3.4]{carlet}  and the references therein. \\
 
Differential uniformity that we now define is a concept, which was introduced three decades ago in \cite{nyberg1} as a measure of resistance of an $(m,n)$ function against the differential attack (see \cite{bihams}).\\

Here we take $m=n$ and use the univariate representation, i.e., we consider mappings   
$G:\F_{2^n} \rightarrow \F_{2^n}$ and study their differential behaviour. The \textit{first derivative of $G$ in the direction} $a\in \F_{2^n}^*= \F_{2^n} \setminus \{0\}$  is defined as $$D_aG(x)=G(x)+G(x+a). $$ 
The cardinality of the set of solutions of $D_aG(x)=b$ for $a\in \F_{2^n}^* $ and $b \in \F_{2^n}$
is denoted by $\delta_G(a,b)$;

$$\delta_G(a,b)=|\{x \in \F_{2^n}: D_aG(x)=b  \}|.$$  
The \textit{differential uniformity} of $G$ is the quantity
$$\delta_G= \max_{a\in \F_{2^n}^*, b \in \F_{2^n}}\delta_G(a,b).$$ 
A function $G$ with differential uniformity $\delta_G$ is said to be differentially $\delta_G$-uniform. \\

The multi-set $\{\delta_G(a,b): a\in \F_{2^n}^*, b \in \F_{2^n}   \}$
is called the \textit{differential spectrum of} $G$.\\

Obviously, $\delta_G\geq 2$, and it is even. The smaller $\delta_G$ is, the more resistant $G$ is against the differential attack (see, for instance, Section 3.4 in \cite{carlet} for details), and the functions attaining the smallest possible value  $\delta_G= 2$  are called \textit{almost perfect nonlinear} (APN) functions.\\

A wide variety of results concerning APN functions can be found in the literature. For comprehensive surveys on the topic, we refer the reader to \cite{blondeau} and \cite{pott}. 
\\

Recently, three new concepts have been used in \cite{papn,pa,vanishing} to better understand the differential behaviour of (non-APN) functions $G:\F_{2^n} \rightarrow \F_{2^n}$. In different ways, these are measures to detect fine differences between functions of the same differential uniformity, and/or to evaluate the distance of $G$ to the set of all APN functions. The essential idea in \cite{papn,vanishing} is to weaken a well-known condition that ensures APN-ness and analyse the implications. Charpin and Kyureghyan, in \cite{pa}, introduce a particular condition to determine APN-ness of $G$, and study the problem of finding sets $S \subsetneq \F_q^*$ that are sufficiently large to assure APN-ness by checking the condition on $S$ only. 
We describe the three concepts, their interrelations, and their relevance to our work in Sections 2, 4 and 5. \\

Clearly, \textquotedblleft distance of $G$ from the set of APN functions\textquotedblright $~$ is a vague notion. One may interpret this distance to be small, for instance, if the differential uniformity of $G$ is small (e.g., $ \leq 6$), and/or if $G$ has many 2-to-1 derivatives. Other than the work in \cite{papn,pa,vanishing} that we mentioned in the previous paragraph, there is eminent research on related questions. The authors of \cite{weight3} use a coding theory view point to study differential properties of non-APN functions and in \cite{localapn}, focusing on the power functions, the concept of \textquotedblleft locally APN-ness\textquotedblright  is intoduced. We point out some connections between this concept and the partial APN-ness, which is defined in \cite{papn}, see Section 5 below, in particular Lemma \ref{locapn1} and Corollary \ref{localapn}. The reader may find recent work on locally APN-ness in \cite{lapnp,locapn}. A critical aspect of resistance of functions against variants of differential cryptanalysis is illustrated in \cite{BCC}. It is pointed out that 
   one may actually need to consider the whole differential spectrum rather than just $\delta_G$, see \cite[Example 3]{BCC} for details.\\ 

Here we follow a similar line of research with the goal of identifying fine differences between non-APN functions (of low differential uniformity). 
We use an alternative approach, the so-called \textit{difference squares}, to acquire comprehensive information about the properties of the derivative functions $D_aG$ %of a function $G$, 
as $a$ varies over $\F_{2^n}^*$. A difference square is simply a table listing the values $D_aG(x)$, for $x \in \F_{2^n}$ and $a \in \F_{2^n}^*$, see Section 2 for the precise definition. Interesting patterns in difference squares, for instance, subtle differences between functions with the same  differential spectrum lead us to introduce a new measure which we call the \textit{ APN-defect}. It is closely related to the previous such measures in \cite{papn,pa,vanishing}, see Sections 2, 4, 5, extends them in a sense that we describe below, and has some appealing features. Most notably, when used for the modifications of the inverse function $F(x)=x^{{2^n}-2}$, the APN-defect differentiates CCZ-inequivalent ones quite efficiently and provides insight into their structural properties that can help constructing  new %\textit{new} 
permutations of  
$\F_{2^n}$.\\

We recall that
modifications of well known functions over finite fields in any characteristic,
in particular of the inverse function or rational fractional permutations $R(x)= (ax+b)/(cx+d), \ ad-bc\neq 0$ with $R(d/c)=a/c$ 
have been studied extensively in relation to a large variety of problems, see 
 the surveys \cite{anbar2,topuzoglu}, and the recent papers  \cite{7,JK22,boomerang,ccds,25}.
The effect of the change of values of APN or differentially 4-uniform functions at a small number of points or in larger sets, especially in subfields, has been of particular interest. Such modifications %, %for instance of the inverse function $x^{q-2}$, 
%not only  
enabled constructions of new functions that have favourable differential properties, high algebraic degrees and high nonlinearity. The study of modified functions has been of interest also because of
its links to challenging problems concerning
upper bounds for algebraic degrees of APN functions, the Hamming distance between them, and the \textit{big APN problem}, i.e., construction of APN permutations of $\F_{2^n}$ for even $n >6$.
 We refer the reader to  \cite{budaghyan,swap,calderini21,kaleyski-thesis,kaleyski,27,33,72,37,43}, and the references therein. \\

This note is structured as follows. After Section 2, where we familiarize the reader with the concepts introduced in \cite{papn,pa,vanishing}, as well as the difference squares and their interrelations, we present a new measure, namely, the APN-defect of a function $G$ in Section 3. We define the concept of \textquotedblleft quasi-APN-ness with respect to APN-defect \textquotedblright and calculate and/or give bounds for the APN-defect of various functions of interest. We show for instance that a non-APN quadratic power function is not quasi-APN. We focus on modifications of the inverse function in Section 4, and we explain why the notion of Carlitz rank is relevant to their study. In this first part of our work we only investigate the case where the inverse function is modified at two points. However, description of the difference square corresponding to this simple case provides insight into our approach that helps to better understand the differential behaviour of functions, especially of those with low differential uniformity, see Corollary \ref{spectf0alpha}, Theorems  \ref{apndeff0a}, \ref{pqsf0a}, Remark \ref{newpqs}. Section 5 is devoted to the analysis of effectiveness of various measures and the relations between them. We end this note with concluding remarks.  \\

\section{Background }

Our first aim is to illustrate the connections between the concepts introduced in \cite{papn,pa,vanishing} and our work. \\

Considering that a function $H:\F_{2^n} \rightarrow \F_{2^n}$ is APN if and only if $D_aH(x)$ is 2-to-1 for all $a \in \F_{2^n}^*$, a natural way of studying differential properties of a non-APN function $G$ is to specify the directions $a$, where $D_aG(x)$ is 2-to-1.

\begin{definition}
	
	\textit{(\cite{pa})
		Let $a \in \F_{2^n}^*$ be fixed. A function $G:\F_{2^n} \rightarrow \F_{2^n}$ is said to satisfy the property $(p_a)$ 
		if the equation
		$D_aG(x)=G(x) + G(x + a) = b$ 
		has either 0 or 2 solutions for every $b \in \F_{2^n}$, i.e., the derivative $D_aG(x)$ of $G$ in direction $a$ is
		2-to-1.}
\end{definition}
Clearly, $G$ is APN if and only if it satisfies the property $(p_a)$ for all $a \in \F_{2^n}^*$. When $G$ is a power function, $G(x)=x^d$ for some fixed integer $d$, one has 
\begin{equation}
\label{monom}
 D_aG(x)=a^d\left((\dfrac{x}{a})^d+(\dfrac{x}{a}+1)^d \right), 
 \end{equation}

\noindent 
and hence  $\delta_G(a,b)=\delta_G(1,b/a^d)$ \textit{for all} $a \in \F_{2^n}^*$. Therefore $G(x)=x^d$ is APN 
if and only if it satisfies the property $(p_1)$. This, of course, is a very special case. However, 
it is shown in \cite{pa} that in fact, it is sufficient to check small sets to ensure APN-ness for many other classes of functions. For instance, 
if all coefficients of $G$ are in $\F_2$, $n$ is odd and $2^n-1$ is a prime, then $G$ is APN if and only if the condition $(p_a)$ is satisfied for  $a \in I,$ where $I$ in an explicitly given set of cardinality $(2^n-2)/2n$, see Corollary 2 in \cite{pa}.
\\

A well-known criterion for APN-ness is the \textit{Janwa-Wilson-Rodier condition} that we state below.\\   

\noindent
\textit{{Janwa-Wilson-Rodier (JWR) Condition} (\cite{JW,rodier}).
A function $G:\F_{2^n} \rightarrow \F_{2^n}$ is APN if and only if all triples of elements
$x$, $y$, $z \in \F_{2^n}$ satisfying
\begin{equation}
	\label{JWR}
	G(x) + G(y) + G(z) + G(x + y + z) =0 
\end{equation}
belong to the surface
$$(x + y)(x + z)(y + z) = 0.$$}

By fixing one of the variables in (\ref{JWR}), say putting $x=x_0$,   
Budaghyan et al. gave the following definition.

\begin{definition}
	%{Definition 1.}
	\textit{(\cite{papn}) Let $x_0 \in \F_{2^n}$ be fixed. A function $G:\F_{2^n} \rightarrow \F_{2^n}$ is defined to be an $x_0$-pAPN function 
		(or $x_0$-partially APN) if all pairs of elements
		$y$, $z \in \F_{2^n}$, satisfying
		\begin{equation*}
					G(x_0) + G(y) + G(z) + G(x_0 + y + z) =0 
		\end{equation*}
		lie on the curve
		$$(x_0 + y)(x_0 + z)(y + z) = 0.$$} 
	
\end{definition}
 \noindent
It is clear therefore that $G$ is APN if and only if it is $x_0$-pAPN for all $x_0 \in \F_{2^n}$. 
\\

Obviously, one can also express the JWR condition in the following equivalent form. \\

A function $G:\F_{2^n} \rightarrow \F_{2^n}$ is APN if and only if any triple of \textit{pairwise distinct} elements
$x$, $y$, $z \in \F_{2^n}$ satisfies

\begin{equation}
	\label{JWR2}
	G(x) + G(y) + G(z) + G(x + y + z) \neq 0.
\end{equation}

Accordingly, Li et al. consider the subset of all 2-dimensional flats 
$\{\{x,y,z,x+y+z\}: x,y,z \in \F_{2^n} {\mathrm{~are~pairwise ~distinct}}\}$ in $\F_{2^n}$, which fail the condition in (\ref{JWR2}), i.e., satisfy (\ref{JWR}).

\begin{definition}
	\cite{vanishing}
	Let $G:\F_{2^n} \rightarrow \F_{2^n}$. The set of vanishing flats with respect to $G$ is defined as
	$${{VF}}_G=\{\{x,y,z,x+y+z\}: x,y,z \in \F_{2^n} {\mathrm{~are~pairwise ~distinct~and~satisfy}}~(\ref{JWR})\}. $$
	
\end{definition}
Therefore, the function $G$ is APN if and only if $VF_G=\emptyset$. \\

Since VF$_G$ is a subset of the set of all 2-dimensional flats in $\F_{2^n}$, the authors of \cite{vanishing} call the system $(\F_{2^n},{{VF}}_G )$
a \textit{partial quadruple system}, which actually is an instance of a ``packing", see \cite{vanishing} for details. This combinatorial relation renders the study of vanishing flats particularly interesting.\\

\subsection{Difference squares}

We use the so-called \textit{difference squares} to analyse the differential behaviour of functions that we study. By fixing an ordering of the elements of $\F_{2^n}$, therefore putting $\F_{2^n}=\{ x_1=0, x_2=1, \ldots, x_{2^n}\}$, we define \textit{the difference square corresponding to the function $G$} to be the $2^n-1$ by $2^n$ array, 
where the $a$-th row $\Delta_a(G)$, $a \in \{ x_2, \ldots, x_{2^n}\}$, consists of the derivatives 
$D_aG(x_1),\ldots,D_aG(x_{2^n}) $. \\

We note that the entries in a difference square are elements of the field (not integers as in a \textit{difference distribution table}: DDT, consisting of entries $DDT_G (a, b) = |\{x \in \F_{2^n}
: D_aG(x)= b\}|$
%  $$
for any  $a, b \in \F_{2^n}$ ).
Difference squares were employed previously in \cite{Ayca,Gary}, in connection with Costas arrays and construction of interleavers, i.e., permutations that are used in turbo codes.\\

In what follows, we often use the following notation.  
$$D_G(a,x)=\{y \in \F_{2^n}: D_aG(y)=D_aG(x)      \},$$
and

$$\nabla_G(a,x)=|D_G(a,x)|=  \delta_G(a,D_aG(x)).$$

When $G$ is clear from the context, we put $\nabla(a,x)=\nabla_G(a,x) $ and we  denote the  $a$-th row of the difference square corresponding to the function $G$ by
%$\Delta_a(G)$, or 
just $\Delta_a$.\\
 
With the above terminology, it is obvious that Definitions 1 and 2 can be expressed as follows. \\

\noindent
{Definition 1*}. \textit{	Let $a \in \F_{2^n}^*$ be fixed. A function $G:\F_{2^n} \rightarrow \F_{2^n}$ satisfies the property $(p_a)$ 
	if any element in the $a$-th row
	$\Delta_a(G)$ appears exactly twice i.e., $\nabla_G(a, x)=2$ for every $x \in \F_{2^n}$.}\\

\noindent	
{Definition 2*}. \textit{Let $x_0 \in \F_{2^n}$ be fixed. A function $G:\F_{2^n} \rightarrow \F_{2^n}$ is $x_0$-pAPN if for each  $a \in \{ x_2, \ldots, x_{2^n}\}$, the element $D_aG(x_0)$ appears exactly twice in the $a$-th row
	$\Delta_a(G)$, i.e., $\nabla_G(a, x_0)=2$ for every $a \in \F_{2^n}^*$.}\\ 

Recalling that $q={2^n}$, we define the \textit{row spectrum} 
$R\mathrm{\textit{-}Spec_{q}}(G)$ and the \textit{column spectrum} $C\mathrm{\textit{-}Spec_{q}}(G)$ as,
$$R\mathrm{\textit{-}Spec_{q}}(G)=\{a \in \F_q^*: G~ \mathrm{satisfies ~the~ property}~ (p_a)   \},$$
and
$$C\mathrm{\textit{-}Spec_{q}}(G)=\{x_0 \in \F_q: G~ \mathrm{is}~ x_0 \mathrm{\textit{-}pAPN}   \}.   $$

Obviously, $G$ is APN if and only if the cardinalities of these spectra satisfy 
$|R\mathrm{\textit{-}Spec_{q}}(G)|=q-1,$ and $|C\mathrm{\textit{-}Spec_{q}}(G)|=q$.\\  

We note that Definitions $1^*$ and $2^*$ enable one to obtain $R\mathrm{\textit{-}Spec_{q}}(G),$ and $C\mathrm{\textit{-}Spec_{q}}(G)$, easily. Indeed, once the elements in the sets 
%\begin{equation} %\label{s}
$A_a= \{D_aG(x_i): 1\leq i\leq {2^n} %~{\mathrm{with}}~ 
,\nabla(a, x_i)>2\} $ 
%\end{equation}
that lie in the a-th row $\Delta_a$ are known, the difference square immediately yields all $a \in \F_{2^n}^*$ and all $x_0 \in \F_{2^n} $,  where $G$ satisfies the property $(p_a)$, and $G$ is $x_0$-pAPN, respectively.\\ 

 Obviously, $a \in R\mathrm{\textit{-}Spec_{2^n}}(G)$ if 
$A_a=\emptyset$. Similarly, $x_0 \in C\mathrm{\textit{-}Spec_{2^n}}(G)$ if 
the column corresponding to $x_0$ does not contain any 
of the elements of $A_a$, which lie in the a-th row $\Delta_a$, as $a$ varies over $\F_{2^n}^*$, see Example 1 below. \\

\begin{example}
	\label{f0alfa}
	Let $n=4$, $\zeta$ be a primitive element of $\F_{16}$. 
	Put $x_1=0, x_i= {\zeta}^{i-2} $ for $ 2 \leq i \leq 16.$
	The difference square of a function $G$  is given below. The elements that are circled are those in the sets $A_a \cap \Delta_a$. 
	Accordingly, one can immediately see $R\mathrm{\textit{-}Spec_{16}}(G)$ and $C\mathrm{\textit{-}Spec_{16}}(G)$. For instance, the row $\Delta _{\zeta ^2}$ is free of the circled entries of the difference square, in other words, 
	$\nabla(\zeta ^2, x_i)=2$ for every $ 1 \leq i \leq 16.$ Therefore, 
	   $\zeta^2 \in R\mathrm{\textit{-}Spec_{16}}(G)$. Similarly, the column corresponding to $x_0=1$ does not have any circled elements, which means that $\nabla(a, 1)=2$ for every $a \in \F_{16}^*$, implying $1 \in C\mathrm{\textit{-}Spec_{16}}(G)$.  Indeed,  
	$R\mathrm{\textit{-}Spec_{16}}(G)=\{\zeta^2, \zeta^3, \zeta^5,  \zeta^9\}$, and $C\mathrm{\textit{-}Spec_{16}}(G)=\{1, \zeta^2, \zeta^3, \zeta^5, \zeta^8, \zeta^9, \zeta^{12}, \zeta^{14}\}$. 
\end{example}
\begin{center}\resizebox{.7\hsize}{!}{  
		\begin{tabular}{l | *{17}{c}}
			
			&${0}$&${1}$& ${\zeta}$&	${\zeta^2}$&	${\zeta^3}$&	${\zeta^4}$&	${\zeta^5}$&	${\zeta^6}$&	${\zeta^7}$&	${\zeta^8}$&\	${\zeta^9}$&	${\zeta^{10}}$&	${\zeta^{11}}$& ${\zeta^{12}}$ & ${\zeta^{13}}$ & ${\zeta^{14}}$  \\\hline
			
			$	{1}$ &$\zeta^{3}$	&$ \zeta^{3 \ }  $&  \mycircled{$\zeta^{11}$}& ${\zeta^{5}}$& ${\zeta^{13}}$& \mycircled{$\zeta^{11}$}& ${1}$& \mycircled{$\zeta^{11}$}& ${\zeta^{14}}$& ${\zeta^{5}}$& ${\zeta^{14}}$& ${1}$& ${\zeta^{7}}$& ${\zeta^{7}}$& \mycircled{$\zeta^{11}$}& ${\zeta^{13}}$ \\
			
			${\zeta}$ &	\mycircled{$\zeta^{14}$}&  $ \zeta^{12}$& \mycircled{$\zeta^{14}$}& ${\zeta^{9}}$& ${\zeta^{4}}$& ${\zeta^{12}}$& ${\zeta^{9}}$& \mycircled{$\zeta^{14}$}& ${\zeta^{10}}$& ${\zeta^{13}}$& ${\zeta^{4}}$& ${\zeta^{13}}$& \mycircled{$\zeta^{14}$}& ${\zeta^{6}}$& ${\zeta^{6}}$& ${\zeta^{10}}$
			\\
			
			$\zeta^2$ & $\zeta^{2}$& 
			$\zeta^{9}$& $\zeta^{10}$& 	$\zeta^{2}$& $\zeta^{8}$& $\zeta^{3}$& 	$\zeta^{10}$& 	$\zeta^{8}$& 	$\zeta^{13}$& 	$\zeta^{9}$& 	$\zeta^{12}$& 	$\zeta^{3}$& $\zeta^{12}$& 	$\zeta^{13}$& 	$\zeta^{5}$& 	$\zeta^{5}$\\

			${\zeta^3}$ & ${\zeta^{5}}$& $ \zeta^{4}$& ${\zeta^{6}}$& ${\zeta^{10}}$& ${\zeta^{5}}$& ${\zeta^{7}}$& ${\zeta^{2}}$& ${\zeta^{10}}$& ${\zeta^{7}}$& ${\zeta^{12}}$& ${\zeta^{6}}$& ${\zeta^{11}}$& ${\zeta^{2}}$& ${\zeta^{11}}$& ${\zeta^{12}}$& ${\zeta^{4}} $\\
			
			${\zeta^4}$ 	& \mycircled{$\zeta^{10}$}&  $1$&  $1$& ${\zeta^{7}}$& ${\zeta^{9}}$& \mycircled{$\zeta^{10}$}& ${\zeta^{6}}$&$ {\zeta}$& ${\zeta^{9}}$& ${\zeta^{6}}$& ${\zeta^{11}}$& ${\zeta^{7}}$&\mycircled{$\zeta^{10}$}& ${\zeta}$& \mycircled{$\zeta^{10}$}& ${\zeta^{11}}$\\
			
			${\zeta^5}$& ${\zeta^{11}}$& $\zeta^{10}$& ${\zeta^{13}}$& ${\zeta^{13}}$& ${\zeta^{6}}$& ${\zeta^{8}}$& ${\zeta^{11}}$& ${\zeta^{5}}$& ${1}$& ${\zeta^{8}}$& ${\zeta^{5}}$& ${\zeta^{10}}$& ${\zeta^{6}}$& ${\zeta^{9}}$& ${1}$& ${\zeta^{9}}$\\
			
			${\zeta^6}$ &	\mycircled{$\zeta^{4}$}& $\zeta^{8}$&	\mycircled{$\zeta^{4}$}& ${\zeta}$&${\zeta}$& ${\zeta^{5}}$& ${\zeta^{7}}$& 	\mycircled{$\zeta^{4}$}&\mycircled{$\zeta^{4}$}& ${\zeta^{14}}$& ${\zeta^{7}}$& \mycircled{$\zeta^{4}$}& 	\mycircled{$\zeta^{4}$}& ${\zeta^{5}}$& ${\zeta^{8}}$& ${\zeta^{14}}$\\
			
			${\zeta^{7}}$ & 	\mycircled{$\zeta^{6}$}& $\zeta^{13}$& ${\zeta}$& ${\zeta^{8}}$& ${1}$& ${1}$& ${\zeta^{4}}$& \mycircled{$\zeta^{6}$}& \mycircled{$\zeta^{6}$}& ${\zeta^{3}}$& ${\zeta^{13}}$& \mycircled{$\zeta^{6}$}& ${\zeta^{3}}$& ${\zeta^{8}}$& ${\zeta^{4}}$& ${\zeta}$ \\

			${\zeta^8}$& ${\zeta}$& $\zeta^{6}$& \mycircled{$\zeta^{5}$}& ${\zeta^{6}}$& ${\zeta^{7}}$& ${\zeta^{14}}$& ${\zeta^{14}}$& ${\zeta^{3}}$& \mycircled{$\zeta^{5}$}& ${\zeta}$& ${\zeta^{2}}$& \mycircled{$\zeta^{5}$}& \mycircled{$\zeta^{5}$}& ${\zeta^{2}}$& ${\zeta^{7}}$& ${\zeta^{3}}$\\

			${\zeta^9}$ 	& ${\zeta^{8}}$& $\zeta^{2}$& ${\zeta^{12}}$& ${\zeta^{11}}$& ${\zeta^{12}}$& ${\zeta^{6}}$& ${\zeta^{13}}$& ${\zeta^{13}}$& ${\zeta^{2}}$& ${\zeta^{4}}$& ${\zeta^{8}}$& ${\zeta}$& ${\zeta^{11}}$& ${\zeta^{4}}$& ${\zeta}$& ${\zeta^{6}}$\\

			${\zeta^{10}}$ 	& \mycircled{$\zeta^{12}$}& $\zeta^{5}$& ${\zeta^{7}}$& ${\zeta^{4}}$& ${\zeta^{10}}$& ${\zeta^{4}}$& ${\zeta^{5}}$& \mycircled{$\zeta^{12}$}& \mycircled{$\zeta^{12}$}& ${\zeta^{7}}$& ${\zeta^{3}}$& \mycircled{$\zeta^{12}$}& ${1}$& ${\zeta^{10}}$& ${\zeta^{3}}$& ${1}$  \\

			${\zeta^{11}}$ & \mycircled{$\zeta^{9}$}& $\zeta^{14}$& \mycircled{$\zeta^{9}$}& ${1}$& ${\zeta^{3}}$& \mycircled{$\zeta^{9}$}& ${\zeta^{3}}$& \mycircled{$\zeta^{9}$}& ${\zeta^{11}}$& ${\zeta^{11}}$& ${1}$& ${\zeta^{2}}$&\mycircled{$\zeta^{9}$}& ${\zeta^{14}}$& \mycircled{$\zeta^{9}$}& ${\zeta^{2 }}$ \\

			${\zeta^{12}}$ &${1}$& $\zeta$& \mycircled{$\zeta^{2}$}& ${\zeta^{3}}$& ${\zeta^{14}}$& \mycircled{$\zeta^{2}$}& ${\zeta^{8}}$& \mycircled{$\zeta^{2}$}& ${\zeta^{3}}$& ${\zeta^{10}}$& ${\zeta^{10}}$& ${\zeta^{14}}$& ${\zeta}$& ${1}$&\mycircled{$\zeta^{2}$}& ${\zeta^{8}}$\\

			${\zeta^{13}}$ &	\mycircled{$\zeta^{13}$}&$\zeta^{7}$& ${\zeta^{3}}$& ${\zeta^{12}}$& ${\zeta^{2}}$& \mycircled{$\zeta^{13}$}& ${\zeta}$& ${\zeta^{7}}$& ${\zeta}$& ${\zeta^{2}}$& ${\zeta^{9}}$& ${\zeta^{9}}$& \mycircled{$\zeta^{13}$}& ${\zeta^{3}}$& \mycircled{$\zeta^{13}$}& ${\zeta^{12}}$\\

			${\zeta^{14}}$& ${\zeta^{7}}$& $\zeta^{11}$& \mycircled{$\zeta^{8}$}& ${\zeta^{14}}$& ${\zeta^{11}}$& ${\zeta}$& ${\zeta^{12}}$& ${1}$& \mycircled{$\zeta^{8}$}& ${1}$& ${\zeta}$& \mycircled{$\zeta^{8}$}& \mycircled{$\zeta^{8}$}& ${\zeta^{12}}$& ${\zeta^{14}}$& ${\zeta^{7}}$ \\
	\end{tabular}}
\end{center}

\begin{remark}
\label{diffsqvsdiffspect}	

It is clear that the difference square of a function $G$ carries significant information about the differential behaviour of $G$, i.e., once it is known, one can immediately determine the differential spectrum. Knowing the differential spectrum of $G$ however, is not sufficient to retrieve the difference square. As in the case of differential spectrum, in general, it is difficult to determine the difference square. However, it is rather straightforward in some cases, see Section 4.\\ 
\end{remark}

\subsection{Difference squares and partial quadruple systems}

The difference square of a function $G$ immediately reveals the vanishing flats, i.e., the partial quadruple system associated to $G$. We first give a toy example to illustrate this. \\

\begin{example}
	\label{inversevf}
	Consider the inverse mapping $F(x)=x^{{2^n}-2}$ over $\F_{16}$. We take  again $\zeta$ to be a primitive element of $\F_{16}$, and have the same ordering as in Example 1, i.e., we put 
 $x_1=0,~ x_i= {\zeta}^{i-2} $ for $ 2 \leq i \leq 16.$ Let $\omega={\zeta}^{5} \in \F_4 \setminus \F_2$, i.e., $\omega^2=\omega+1$. It is straightforward to see that $\delta_F=4$ and 
\begin{equation}
\label{invbasic} 
 D_aF(0)=D_{a}F(a\omega)=a^{-1}
 \end{equation} 
  for any $ a \in \F_{16}^*$. In other words, $D_F(a,0)=D_F(a,a\omega )=\{0,a,a\omega,a\omega+a \}$. Therefore, corresponding to each $ a \in \F_{16}^*$, there is exactly one vanishing flat $\{0,a,a\omega,a\omega^2\}$. These vanishing flats can be identified immediately in the difference square given below, where the case of $a=1$ is marked by circles.
  
   On the other hand, corresponding to the values of $a; \ a_1={\zeta}^{i}$, $a_2={\zeta}^{i}\omega$, $a_3={\zeta}^{i}\omega^2$ with $a_1+a_2+a_3=0$ (or to the rows $\Delta_{{\zeta}^{i}}(F)$, $\Delta_{{\zeta}^{i}\omega}(F)$, $\Delta_{{\zeta}^{i}\omega^2}(F)$, $ 0 \leq i \leq 4$), the vanishing flats coincide, so that there are exactly $5$ distinct vanishing flats;
 $\{0,1,{\zeta}^{5},{\zeta}^{10}\},\  \{0,{\zeta},{\zeta}^{6},{\zeta}^{11}\}$,
 $\{0,{\zeta}^2,{\zeta}^{7},{\zeta}^{12}\}$, $\{0,{\zeta}^3,{\zeta}^{8},{\zeta}^{13}\}$, $\{0,{\zeta}^4,{\zeta}^{9},{\zeta}^{14}\}$. In other words, the partial quadruple system associated to $F$ is $(\F_{16}, VF_F)$, where
 $$VF_F=\{\{0, {\zeta}^{i},{\zeta}^{i}\omega,{\zeta}^{i}\omega^2 \}, 0\leq i\leq 4 \}.$$
  
		\begin{center}\resizebox{.7\hsize}{!}{  
				\begin{tabular}{l | *{17}{c}}
					
					&\mycircled{${0}$}&\mycircled{${1}$}&	${\zeta}$&	${\zeta^2}$&	${\zeta^3}$&	${\zeta^4}$&\mycircled{	${\zeta^5}$}&	${\zeta^6}$&	${\zeta^7}$&	${\zeta^8}$&	${\zeta^9}$&\mycircled{	${\zeta^{10}}$}&	${\zeta^{11}}$& ${\zeta^{12}}$ & ${\zeta^{13}}$ & ${\zeta^{14}}$   \\\hline
					
						$	{	1} $&\mycircled{	{{1}}}& \mycircled{	{{1}}}& ${\zeta^{10}}$ &$ {\zeta^{5}}$& ${\zeta^{13}}$& ${\zeta^{10}}$&\mycircled{ {{1}}}& ${\zeta^{11}}$& ${\zeta^{14}}$& ${\zeta^5}$& ${\zeta^{14}}$ &\mycircled{	{{1}}}& ${\zeta^7}$& ${\zeta^7}$& ${\zeta^{11}}$&$ {\zeta^{13}}$\\
					
									${\zeta}$ &			${\zeta^{-1}}	$& 	${\zeta^{12}}	$& 		$ 	{\zeta^{-1}}	$& 	${\zeta^9}	$&	$ {\zeta^4}	$& 	${\zeta^{12}}	$& 	${\zeta^9}	$& 		$	{\zeta^{-1}}	$& 	${\zeta^{10}}	$& 	${\zeta^{13}}	$& 	${\zeta^4}	$& 	${\zeta^{13}}	$& 		${\zeta^{-1}}	$& 	${\zeta^6}	$& 	${\zeta^6}	$& 	${\zeta^{10}}	$
					\\
					
									${\zeta^2}$ &	${\zeta^{-2}}$& ${\zeta^9}$& ${\zeta^{11}}$&	 ${\zeta^{-2}}$& ${\zeta^8}$& ${\zeta^3}$& ${\zeta^{11}}$& ${\zeta^{8}}$&	 ${\zeta^{-2}}$& ${\zeta^{9}}$& ${\zeta^{12}}$& ${\zeta^{3}}$& ${\zeta^{12}}$&	 ${\zeta^{-2}}$& ${\zeta^{5}}$& ${\zeta^{5}}$ \\

									${\zeta^3}$ &	 ${\zeta^{-3}}$& ${\zeta^4}$& ${\zeta^8}$& ${\zeta^{10}}$&	 ${\zeta^{-3}}$& ${\zeta^7}$& ${\zeta^2}$& ${\zeta^{10}}$& ${\zeta^7}$&	 ${\zeta^{-3}}$& ${\zeta^8}$& ${\zeta^{11}}$& ${\zeta^2}$& ${\zeta^{11}}$&	 ${\zeta^{-3}}$& ${\zeta^4}$\\
					
									${\zeta^4}$ 	&	 ${\zeta^{-4}}$& ${\zeta^{3}}$& ${\zeta^{3}}$& ${\zeta^{7}}$& ${\zeta^{9}}$&	 ${\zeta^{-4}}$& ${\zeta^{6}}$& ${\zeta}$& ${\zeta^{9}}$& ${\zeta^{6}}$&	 ${\zeta^{-4}}$& ${\zeta^{7}}$& ${\zeta^{10}}$& $ {\zeta}$& ${\zeta^{10}}$&	 ${\zeta^{-4}}$\\
					
									${\zeta^5}$&	 ${\zeta^{-5}}$&	 ${\zeta^{-5}}$& ${\zeta^{2}}$& ${\zeta^{2}}$& ${\zeta^{6}}$& ${\zeta^{8}}$&	 ${\zeta^{-5}}$& ${\zeta^{5}}$& ${1}$& ${\zeta^{8}}$& ${\zeta^{5}}$&	 ${\zeta^{-5}}$& ${\zeta^{6}}$& ${\zeta^{9}}$&  ${1}$& ${\zeta^{9}}$ \\
					
									${\zeta^6}$ &	${\zeta^{-6}}$& ${\zeta^{8}}$&	 ${\zeta^{-6}}$& ${\zeta}$& ${\zeta}$& ${\zeta^{5}}$& ${\zeta^{7}}$&	 ${\zeta^{-6}}$& ${\zeta^{4}}$& ${\zeta^{14}}$& ${\zeta^{7}}$& ${\zeta^{4}}$&	 ${\zeta^{-6}}$& ${\zeta^{5}}$& ${\zeta^{8}}$& ${\zeta^{14}}$\\
					
								${\zeta^{7}}$ &	 ${\zeta^{-7}}$& ${\zeta^{13}}$& ${\zeta^{7}}$&	 ${\zeta^{-7}}$& ${1}$& ${1}$& ${\zeta^{4}}$& ${\zeta^{6}}$&	 ${\zeta^{-7}}$& ${\zeta^{3}}$& ${\zeta^{13}}$& ${\zeta^{6}}$& ${\zeta^{3}}$&	 ${\zeta^{-7}}$& ${\zeta^{4}}$& ${\zeta^{7}}$\\

								${\zeta^8}$& ${\zeta^{-8}}$&	 ${\zeta^{6}}$& ${\zeta^{12}}$& ${\zeta^{6}}$&	 ${\zeta^{-8}}$& ${\zeta^{14}}$& ${\zeta^{14}}$& ${\zeta^{3}}$& ${\zeta^{5}}$&	 ${\zeta^{-8}}$& ${\zeta^{2}}$& ${\zeta^{12}}$& ${\zeta^{5}}$& ${\zeta^{2}}$&	 ${\zeta^{-8}}$& ${\zeta^{3}}$ \\

							${\zeta^9}$ &	${\zeta^{-9}}$& ${\zeta^{2}}$& ${\zeta^{5}}$& ${\zeta^{11}}$& ${\zeta^{5}}$& ${\zeta^{-9}}$& ${\zeta^{13}}$& ${\zeta^{13}}$& ${\zeta^{2}}$& ${\zeta^{4}}$& ${\zeta^{-9}}$& ${\zeta}$& ${\zeta^{11}}$& ${\zeta^{4}}$& ${\zeta}$& ${\zeta^{-9}}$\\

							${\zeta^{10}}$ &	${\zeta^{-10}}$& ${\zeta^{-10}}$& ${\zeta}$& ${\zeta^{4}}$& ${\zeta^{10}}$& ${\zeta^{4}}$& ${\zeta^{-10}}$& ${\zeta^{12}}$& ${\zeta^{12}}$& ${\zeta}$& ${\zeta^{3}}$& ${\zeta^{-10}}$& ${1}$& ${\zeta^{-10}}$& ${\zeta^{3}}$& ${1}$ \\

								${\zeta^{11}}$ &${\zeta^{-11}}$& ${\zeta^{14}}$& ${\zeta^{-11}}$& ${1}$& ${\zeta^{3}}$& ${\zeta^{9}}$& ${\zeta^{3}}$& ${\zeta^{-11}}$& ${\zeta^{11}}$& ${\zeta^{11}}$& ${1}$& ${\zeta^{2}}$& ${\zeta^{-11}}$& ${\zeta^{14}}$& ${\zeta^{9}}$& ${\zeta^{2}}$ \\

							${\zeta^{12}}$ &${\zeta^{-12}}$& ${\zeta}$& ${\zeta^{13}}$& ${\zeta^{-12}}$& ${\zeta^{14}}$& ${\zeta^{2}}$& ${\zeta^{8}}$& ${\zeta^{2}}$& ${\zeta^{-12}}$& ${\zeta^{10}}$& ${\zeta^{10}}$& ${\zeta^{14}}$& ${\zeta}$& ${\zeta^{-12}}$& ${\zeta^{13}}$& ${\zeta^{8}}$\\

							${\zeta^{13}}$ &	${\zeta^{-13}}$& ${\zeta^{7}}$& ${1}$& ${\zeta^{12}}$& ${\zeta^{-13}}$& ${\zeta^{13}}$& ${\zeta}$& ${\zeta^{7}}$& ${\zeta}$& ${\zeta^{-13}}$& ${\zeta^{9}}$& ${\zeta^{9}}$& ${\zeta^{13}}$& ${1}$& ${\zeta^{-13}}$& ${\zeta^{12}}$\\

								${\zeta^{14}}$& ${\zeta^{-14}}$& ${\zeta^{11}}$& ${\zeta^{6}}$& ${\zeta^{14}}$& ${\zeta^{11}}$& ${\zeta^{-14}}$& ${\zeta^{12}}$& ${1}$& ${\zeta^{6}}$& ${1}$& ${\zeta^{-14}}$& ${\zeta^{8}}$& ${\zeta^{8}}$& ${\zeta^{12}}$& ${\zeta^{14}}$& ${\zeta^{-14}}$ \\
			\end{tabular}}
		\end{center}

\end{example}

\begin{remark} 
	\label{reminv}  
We note that the argument used in Example \ref{inversevf} is independent of the field $\F_{2^n}$, as long as $n$ is even. Equation (\ref{invbasic}) is well known to hold for any even $n$, with $\omega\in \F_4 \setminus \F_2$, and that $\nabla(a,x)=2$, for any $ x \notin D_F(a,0)$, $a \in \F_{2^n}^*$.	
\end{remark}

%$x \in \F_{2^n}$. 
By the above remark, when $n$ is an arbitrary even integer, one only needs to take a primitive element $\zeta$ of $\F_{2^n}$ and put $\omega={\zeta}^{\frac{2^n-1}{3}}$, implying $<{\zeta}^{\frac{2^n-1}{3}} >= \F^*_4$, and the next lemma follows exactly as in Example \ref{inversevf}. 

\begin{lemma}
	\label{invlemma}
Let $F(x)=x^{2^n-2}$ and $n$ be even. Put $\omega={\zeta}^{\frac{2^n-1}{3}}$ for a primitive element $\zeta$ of $\F_{2^n}$.The partial quadruple system associated to $F$ is $(\F_{2^n}, VF_F)$, where 
$$VF_F=\{\{0, {\zeta}^{i},{\zeta}^{i}\omega,{\zeta}^{i}\omega^2 \}, 0\leq i\leq \frac{2^n-4}{3} \}.$$
Hence, the number of vanishing flats with respect to $F$ is given as $|VF_F|=\frac{2^n-1}{3}.$

\end{lemma}

Lemma \ref{invlemma} is given as a part of Theorem III.3. in \cite{vanishing}, where the proof is presented in a slightly different terminology then that we use above. \\

	The following lemma with a constructive proof is easy to obtain. We include it here in order to illustrate how easily one can acquire the structural information about $G$ through difference squares, in particular, for detecting partial quadruple systems. We use this result in Section 4 to obtain the partial quadruple system associated to a modification of the inverse function, see Theorem \ref{pqsf0a} below.
	
\begin{lemma}
	\label{pqs}
	
	 Let $G:\F_{2^n} \rightarrow \F_{2^n}$ be a non-APN function. The partial quadruple system associated to $G$ is determined by the difference square corresponding to $G$. 
	\end{lemma}
	\begin{proof} Since $G$ is non-APN, there exist elements $a \in \F_{2^n}^*$ and
		$x \in \F_{2^n}$ such that $\nabla(a, x) \geq 4$.
			
		Firstly, suppose that $\nabla(a, x) =4$ for some $a, x$. This means the set $D_G(a,x)=\{y \in \F_{2^n}: D_{a}G(y)=D_{a}G(x)      \}$ has four distinct elements forming the vanishing flat $\{x, x+a,y,y+a\}$, which can be immediately identified knowing the difference square.
		 
		If $\nabla(a, x) =2k, \ k >2$ for some $a\in \F_{2^n}^*,\   x\in \F_{2^n}$, then 		
		there are $k$ distinct pairs $(y,y+a) \in D_G(a,x) \times D_G(a,x)$. %	where  $y,y+a \in D_G(a,x)$.		
		Taking any two such distinct pairs together at a time, one obtains all the  $ k \choose 2$  vanishing flats $\{y_1, \ y_1+a, \  y_2, \ y_2+a \}$, $y_1, y_2 \in D_G(a,x)$, corresponding to the element $x$  and the row $\Delta_{a}$ of the difference square. One can therefore obtain all vanishing flats in this way, as $x$ and $a$ vary over $\F_{2^n}$, $a \neq 0$. However, some of them coincide (in fact, each occurs exactly three times, see for instance, proof of Theorem II.3 in \cite{vanishing}).  Indeed, when $a_1+a_2+a_3=0$ for $a_1,a_2,a_3 \in \F_{2^n}^*$, the vanishing flats $\{x, x+a_i,y,y+a_i\}, i=1,2$ and $\{x,x+a_3,z,z+a_3\}$ coincide for $x+a_1=y+a_2=z$ (in fact, they all are the same as $\{x, x+a_1, x+a_2, x+a_3\})$. 		 
$\hfill\square$
		\end{proof}

  	\begin{example}
  		\label{vffoa}
  Consider the function $G$ in Example \ref{f0alfa} and the difference square corresponding to it. One can see that  $\nabla(1, \zeta) = \nabla(\zeta,0)=\nabla(\zeta^4,0)=\nabla(\zeta^7,0)= \nabla(\zeta^8,\zeta)=\nabla(\zeta^{10},0)= \nabla(\zeta^{12},0)=\nabla(\zeta^{13},0)=\nabla(\zeta^{14},\zeta)=4$. Note that $D(1,\zeta)=\{\zeta, \zeta^4, \zeta^6, \zeta^{13}  \}$ and therefore $\{\zeta, \zeta^4, \zeta^6, \zeta^{13}  \}$ is a vanishing flat.   
 Similarly, corresponding to the values of $a$ and $x$ listed above, the vanishing flats $\{0, \zeta, \zeta^6, \zeta^{11}\}$, $\{0, \zeta^4, \zeta^{11},\zeta^{13}\}$, $\{0, \zeta^6, \zeta^{7},\zeta^{10}\}, \{\zeta,\zeta^7, \zeta^{10},\zeta^{11}\}$ can be immediately identified on the difference square. On the other hand, $\nabla(\zeta^6,0) = \nabla(\zeta^{11},0)=6.$ Therefore, corresponding to the values $a=\zeta^6, x=0$, and $a=\zeta^{11},x=0$, there are ${3 \choose 2}=3 $ vanishing flats, respectively,  which can again be observed using the difference square. Indeed, we have the vanishing flats $\{0, \zeta^6, \zeta, \zeta+\zeta^6=\zeta^{11}\}$, $\{0, \zeta^6, \zeta^7, \zeta^7+ \zeta^6=\zeta^{10}\}$, $\{\zeta, \zeta^{11},\zeta^7, \zeta^{10}\}$ for $a=\zeta^6, x=0$ and similarly, $\{0, \zeta^{11}, \zeta, \zeta^6\}$,
  $\{0, \zeta^{11}, \zeta^4, \zeta^{13}\}$, $\{\zeta, \zeta^6,\zeta^4, \zeta^{13}\}$ for $a=\zeta^{11},x=0$. One then has the partial quadruple system associated to $G$, and the the number of vanishing flats, which is $|VF_G|=5$. By Lemma 1, this value is the same for the inverse function over $\F_{2^4}$.
  
   	\end{example}

  	\begin{remark} As can be seen in the proof of Lemma \ref{pqs}, \textquotedblleft knowing \textquotedblright
  		 the difference square of a function $G$ is equivalent to knowing the values 
  $\nabla_G(a, x)$ and the sets $D_G(a,x)$ for all $a \in \F_{2^n}^*$ and
  $x \in \F_{2^n}$. As mentioned earlier, one cannot expect to have this knowledge in general, but we shall show later that this information can be obtained rather easily for some particular functions of interest to us. 	 
  		
  	\end{remark}

\section{A new measure: APN-defect}

We now propose a new measure. We first describe the properties of this measure and later, in Section 5, clarify our motivation for introducing it by pointing to some of its 
advantages over the previous ones that we described in Section 2.\\

Consider the two subsets of $\F_{2^n}$ defined as,
\begin{equation}
\label{sa}
S_a= \{x \in \F_{2^n}: \nabla_G(a,x)= 2 \}, ~~~
S_a^c= \{x \in \F_{2^n}: \nabla_G(a,x)> 2 \}, 
\end{equation}
for a given $a \in \F_{2^n}^* $.\\

Suppose that the image of the set $S_a^c$ under the difference map 
$D_aG$ has $t_a$ elements;
$$ |{\rm{Im}}_{D_aG}(S_a^c)  |=|\{D_aG(x): x\in S_a^c \}|=t_a. $$
Put
${\rm{Im}}_{D_aG}(S_a^c)=\{b_1, \ldots, b_{t_a} \}$.
Then there exist $k_1^{(a)}, k_2^{(a)}, \ldots,k_{t_a}^{(a)}\in \mathbb{Z}$ such that
$\delta_G(a, b_i)=2k_{i}^{(a)}$, $k_{i}^{(a)} \geq 2 $, for 
$1 \leq i\leq t_a$.
Therefore,
\begin{equation}
\label{sac}	
|S_a^c|=\sum_{i=1}^{t_a}2k_i^{(a)}. 
\end{equation}
In what follows we consider the weighted sums
\begin{equation} 
\label{weightedsac}	
w_{S_a^c}=\sum_{i=1}^{t_a}\left(\frac{k_i^{(a)}}{2}\right)2k_i^{(a)}=
\sum_{i=1}^{t_a}\left(k_i^{(a)}\right)^2.
\end{equation}
We also put
$\chi_a=1  $ if $|S_a|=2^n$, and $\chi_a=0  $ if $|S_a|<2^n$. In other words, $\chi_a=1  $ if and only if $a \in R\mathrm{\textit{-}Spec_{2^n}}(G)$.\\

We are now ready to define 
\begin{equation}
\label{dg}
{\cal{D}}(G)=\sum_{a\in \F_{2^n}^*}(|S_a|-w_{S_a^c}+\chi_a). 
\end{equation}
Note that $G$ is APN if and only if ${\cal{D}}(G)=q^2-1$, $q={2^n}$. Accordingly we define the \textit{APN-defect of $G$} as,
$$ \mathrm{APN\textit{-}def}(G)=q^2-1-{\cal{D}}(G).$$

In order that a function $G$ is favourable with respect to $\mathrm{APN\textit{-}def}(G)$ one would expect $|S_a|$, i.e., the number of $2$-to-$1$ derivatives and the $R\mathrm{\textit{-}Spec_{q}}(G)$ to be large (when compared to $k$-to-$1$ derivatives, $k>2$). Therefore, we call $G$ to be \textit{quasi-APN with respect to $\mathrm{APN\textit{-}def}(G)$ if ${\cal{D}}(G)>0$}. 
\begin{remark}
For a quasi-APN function $G$, the value of the ratio
$R(G)=\mathcal{D}(G)/(q^2-1)$ is of interest. Arguably, the larger the value of $R(G)\leq 1$ is, the closer $G$ is to being APN. 	
\end{remark}

\begin{remark}
Our choice of considering the weighted sum $w_{S_a^c}$
in the definition of APN-defect, rather than $|S_a^c|$, is possibly clear to the reader; we note that $|S_a|+|S_a^c|=q$. 
\end{remark} 
 
 Obviously, $\mathrm{APN\textit{-}def}(G)$ can be expressed in terms of the quantities
$\nabla_G(a, x)$ easily, since
$$|S_a|=\sum_{x\in \F_{2^n}}{2 \choose \nabla_G(a,x)},$$ %~\mathrm{and}~ %|wS_a^c|=\sum_{x\in \F_{2^n}}\frac{2}{(\nabla_G(a,x)-2)^*}{\nabla_G(a,x)/2 \choose 2},$$\\
$$\chi_a=\lfloor {\frac{1}{{2^n}}|S_a|}\rfloor , $$
and
$$w_{S_a^c}=\frac{1}{4}\sum_{x\in \F_{2^n},\nabla_G(a,x) > 2 }{\nabla_G(a,x)}=\frac{1}{4}\sum_{x\in \F_{2^n}}{\nabla_G(a,x)\left(1-{2 \choose \nabla_G(a,x)}\right)}.$$\\
Hence,
%\label{lastspectrum}
$$	{\cal{D}}(G)=\sum_{a\in \F_{2^n}^*}\left(\sum_{x\in \F_{2^n}}{2 \choose \nabla_G(a,x)} - \frac{1}{4}\sum_{x\in \F_{2^n}}{\nabla_G(a,x)\left(1-{2 \choose \nabla_G(a,x)}\right)} + \chi_a \right)  $$

%$$=\sum_{a\in \F_{2^n}^*}\left(\sum_{x\in \F_{2^n}}\left({2 \choose \nabla_G(a,x)} - \frac{1}{4} {\nabla_G(a,x)\left(1-{2 \choose \nabla_G(a,x)}\right)}\right)
% + \chi_a \right)  $$

%\begin{equation} %\label{lastspectrum}
$$=\sum_{a\in \F_{2^n}^*}\left(\sum_{x\in \F_{2^n}}\left({2 \choose \nabla_G(a,x)}\left(1 + \frac{1}{4} {\nabla_G(a,x)}\right)-\frac{1}{4}{\nabla_G(a,x)}\right)
+ \chi_a \right) .$$ 
%\end{equation}

As we have mentioned earlier, see Remark \ref{diffsqvsdiffspect}, a difference square carries more information about the derivatives than the differential spectrum of a function. However, $	{\cal{D}}(G)$ can also be expressed in terms of the quantities $\delta_G(a,b)$ for $ a\in \F_{2^n}^*, b \in \F_{2^n}$, since $	{\cal{D}}(G)$ involves {\it{sums} over all $a\in \F_{2^n}^*$ and all $x \in \F_{2^n}$}. Indeed, we have
\begin{equation} 
\label{dela}
|S_a|=2 \sum_{b\in \F_{2^n}, \delta_G(a,b)>0}{2 \choose \delta_G(a,b)}
= \sum_{b\in \F_{2^n}}{2 \choose \delta_G(a,b) }\delta_G(a,b),
\end{equation}

and

\begin{equation} 
\label{delwa}
w_{S_a^c}=\sum_{b\in \F_{2^n},\delta_G(a,b) > 2 }{\left(\frac{\delta_G(a,b)}{2}\right)^2}=\frac{1}{4}\sum_{b\in \F_{2^n}}{\left(\delta_G(a,b)\right)^2\left(1-{2 \choose \delta_G(a,b)}\right)}.
\end{equation}

Therefore, the $\mathrm{APN\textit{-}def}(G)$ can be evaluated if the differential spectrum of $G$ is known. For instance, Table 1 below gives the APN-defect of power functions, differential spectra of which are known.

\subsection{APN-defect of some special functions}

Now we obtain bounds for the APN-defect of some special functions that are not APN. %where $q=2^n$.
\begin{theorem}
	\label{power}
	Let $G(x)=x^d$ be a non-APN power function, $q=2^n$. Then
	\begin{align} 
	\label{defect-mono}
		9(q-1)     \leq      \mathrm{APN\textit{-}def}(G) \leq  \frac{1}{4}(q^2+4q+4)(q-1).
	\end{align}
	Moreover, the lower bound is attained by the inverse function (when $n$ is even). %$\mathrm{APN\textit{-}def}(F)$ attains the upper bound if and only if $F$ is a linear monomial.
\end{theorem}
\begin{proof}
Recall that 
 	 the multiset 
	$\{ \delta_G(a,x) : x\in \mathbb{F}_q\}$ is the same for each nonzero $a$, see Equation (\ref{monom}). Then our assumption that $G$ is not APN implies that
	$\sum_{a \in \F_q^*}\chi_a=0$. Therefore there exist positive integers $\ell, t, k_1, \ldots ,k_t$ such that 
	\begin{align*}
		|S_1|=2\ell \quad \text{and} \quad |S_1^c|=\sum_{i=1}^{t}2k_i, 
			\end{align*}
	see (\ref{sac}) above. Hence, we have 
	\begin{align}
	\label{mono}
		{\cal{D}}(F) = (q-1)\left( 2\ell-\sum_{i=1}^{t}k_i^2\right).
	\end{align}
	Then by Equation (\ref{mono}),
	\begin{align}
	\label{bound-mono}
		- \frac{(q-1)q^2}{4} \leq {\cal{D}}(F) \leq (q-1)((q-4)-4)=(q-1)(q-8),
	\end{align}
	and equality in the lower bound holds if and only if $G$ is linear.
	Note that 
		$G(x)=x^d$ attains the 
		upper bound in Equation (\ref{bound-mono}) if and only if $2\ell=2^n-4$, i.e., whenever $t=1$, and $2k_1=4$. 
	Therefore, $F(x)=x^{q-2}$ attains the lower bound in Equation \eqref{defect-mono}.  
	$\hfill\square$
\end{proof}

\begin{remark}
\label{remark4}
	We have shown above that the APN-defect of $F(x)=x^{q-2}$ is $9(q-1)$, the minimum value that a non-APN power function can have. This is expected of course, since it has the maximum number of 2-to-1 derivatives among such functions, see Remark \ref{reminv}. 
\end{remark}
	
	As mentioned earlier, it is possible to find the exact value of APN-defect of any power function $x^d$ when the differential spectrum of $x^d$ is known.  Table 1 below presents the values of ${\cal{D}}(x^d)$ for such functions. Compare it with Table III.I in \cite{vanishing}, which gives the number of vanishing flats.  We use the same notation as in  
	\cite{vanishing}. One has $\omega_i=\ell_{2i}/(2^n-1)$, where $\ell_{2i}$ is the frequency of $2i$ in the differential spectrum
	$\{\delta_{x^d}(a,b): a\in \F_q*, ~b \in \F_q  \}$, and $K$ denotes the Kloosterman sum with the explicit expression 
	$
	K=1+\frac{(-1)^{n-1}}{2^{n-1}}\sum_{i=0}^{ \frac{n}{2} } (-1)^i \binom{n}{2i}7^i. 
	$ We put  $\Delta(a,b)=	1 $ when $a \mid b$, and
	$\Delta(a,b)=0 $ when $a \nmid b$.

\begin{table*}\renewcommand{\arraystretch}{1}
	\begin{center}
	\caption*{{Table 1: }Power functions $x^d$ over $\F_{2^n}$ with known differential spectra,  and corresponding values of $\mathcal{D}(x^d)$ %(Compare with Table III.I in \cite{vanishing})
		}
	\resizebox{!}{.20\paperheight}{%	
		
		\begin{tabular}{|c|c|c|c|c|}
					\hline
			$n$ & $d$ &  Differential Spectrum &$\mathcal{D}(x^d)$ & Reference \\ \hline
			\multirow{3}{*}{$n \ge 2$} & $2^t+1$ &  $w_0=2^n-2^{n-s}$ & \multirow{3}{*}{$-(2^n - 1)2^{n + s - 2}$} & \multirow{3}{*}{\cite[Section 5.2]{BCC}}  \\ 
			& $1 \le t \le n/2,$ &$w_{2^s}=2^{n-s}$ & &  \\
			&$gcd(n,t)=s$ &&&\\
			 \hline
			$n \ne 3t$ & $2^{2t}-2^t+1$ &  $w_0=2^n-2^{n-s}$ & \multirow{3}{*}{$-(2^n - 1)2^{n + s - 2}$} &  \multirow{3}{*}{\cite[Theorem 2]{BCC}} \\
			$n/s$ odd & $2 \le t \le n/2,$ & $w_{2^s}=2^{n-s}$ & &  \\
			&  $gcd(n,t)=s$ &&&\\ 
			\hline
			\multirow{3}{*}{$n$ even}& \multirow{3}{*}{$2^n-2$} & $w_0=2^{n-1}+1$ & \multirow{3}{*}{$(2^{n} - 1)(2^n - 8)$} &  \multirow{3}{*}{\cite[Proposition 6]{nyberg1}}\\
			& & $w_{2}=2^{n-1}-2$ & &  \\
			& &  $w_{4}=1$ &  & \\ \hline
			\multirow{3}{*}{$n=4t$}& \multirow{3}{*}{$2^{2t}+2^t+1$}  & $w_0=5\cdot2^{n-3}-2^{3t-3}$ & \multirow{3}{*}{$(2^{n}-  1)2^{3t}$} & \multirow{2}{*}{\cite[Example 4]{BCC}} \\
			& & $w_{2}=2^{n-2}+2^{3t-2}$ &  & \\
			&  & $w_{4}=2^{n-3}-2^{3t-3}$ & &  \cite[Theorem 1]{XY} \\ \hline
			\multirow{4}{*}{$n \ge 6$}&\multirow{4}{*}{$7$} & $w_0=2^{n-1}+2w_6+w_4$ & \multirow{4}{*}{$(2^{n} - 1)[2^n - 2^{n-2} - 1 -\frac{13\Delta(2, n)}{6}$}  & \\
			& & $w_{2}=2^{n-1}-3w_6-2w_4$ &   &  \\
			& & $w_{4}=\Delta(2,n)$  & & \cite[Theorem 5]{localapn} \\
			& & $w_{6}=\frac{2^{n-2}+1-5w_4}{6}+(-1)^n\frac{K}{8}$ & {$-\frac{9\cdot 2^{n - 2} + 9 }{6} - (-1)^n\frac{15K}{8}]$} &  \\ \hline
			\multirow{5}{*}{$n \ge 6$}& $2^{n-2}-1$ &  $w_0=2^{n-1}+2w_6+3w_8$ & &   \\
			& or &  $w_{2}=2^{n-1}-3w_6-4w_8$ & $(2^n - 1)[2^n -14\Delta(3,n)$ &  \cite[Corollary 5]{localapn} \\
			& $2^{\frac{n-1}{2}}-1$ & $w_{6}=\frac{2^{n-1}-3-(-1)^n5}{12}$ &{$\frac{-10 \cdot 2^{n - 1} + 30 + (20-15K)(-1)^n}{8}]$}  &  \cite[Theorem 5]{BP} \\
			& $n$ odd &  $+(-1)^n\frac{K}{8}-w_8$ &  & \\
			&  & $w_{8}=\Delta(3,n)$ & & \\ \hline
			& \multirow{4}{*}{$2^{\frac{n}{2}}-1$} &  $w_0=2^{n-1}+2^{\frac{n}{2}-1}-2+w_4$ & \multirow{4}{*}{$(2^n - 1)(2^n- 2^{n - 2} - 4 + 4\Delta(4, n) + 1)$} & \multirow{4}{*}{\cite[Theorem 7]{localapn}}  \\
			$n \ge 6$ &  & $w_{2}=2^{n-1}-2^{\frac{n}{2}-1}+1-2w_4$ &  &  \\
			$n$ even & & $w_{4}=1-\Delta(4,n)$ &  & \\
			& & $w_{2^{\frac{n}{2}}-2}=1$ &  & \\ \hline
			& \multirow{3}{*}{$2^{\frac{n}{2}+1}-1$} & $w_0=2^{n-1}+2^{\frac{n}{2}-1}-1$ & \multirow{3}{*}{
				$2^{2 n} - 2^{3  n/2} - 2^{2  n - 2} - 2^n + 2^{n/2} + 2^{n - 2}$} & \multirow{3}{*}{\cite[Theorem 8]{localapn}}\\
			$n \ge 6$  &  & $w_{2}=2^{n-1}-2^{\frac{n}{2}-1}$ & &    \\
			$n$ even & & $w_{2^{\frac{n}{2}}}=1$ & &  \\
			\hline
			& \multirow{3}{*}{$2^{\frac{n+3}{2}}-1$} & $w_0=2^{n-1}+2w_6+2\Delta(3,n)$ & \multirow{3}{*}{$(2^n - 1)[2^n - 2^{n - 2} - 1- 6\Delta(3, n)-\frac{9\cdot2^{n - 2} + 9}{2} + \frac{15K}{8}]$} & \multirow{3}{*}{\cite[Theorems 1,5]{BP}}  \\
			$n \ge 7$  & & $w_{2}=2^{n-1}-3w_6-3\Delta(3,n)$ & &  \\
			$n$ odd & &  $w_6=\frac{2^{n-2}+1}{6}-\frac{K}{8}$ & & \\ \hline
			&  & $w_0=89 \cdot 2^{n-7}+7\cdot 2^{t-7}(4-K)$ & \multirow{5}{*}
			{$-(4 - K)(5\cdot 2^{t - 4} + 36 \cdot 2^{t - 6} - 16\cdot2^{t - 7}- 9\cdot 2^{t - 5})]$}	& 
			
			\\
			$n=2t$ & $2^{t+1}+2^{\frac{t+1}{2}}+1$ &  $w_{2}=5 \cdot 2^{n-5}-5 \cdot 2^{t-5}(4-K)$ & {$(2^n - 1)[5\cdot2^{n - 4} - 28\cdot2^{n - 6}-9\cdot2^{n - 5} - 16\cdot 2^{n - 7} $} & \\
			$t \ge 5$ & or &  $w_{4}=7 \cdot 2^{n-6}+9 \cdot 2^{t-6}(4-K)$ &  &\cite[Theorem 1.4]{XYY}\\
			$t$ odd & $2^{t+1}+3$ &  $w_{6}=2^{n-5}-2^{t-5}(4-K)$ & &\\
			& & $w_{8}=2^{n-7}-2^{t-7}(4-K)$ && \\ \hline
	\end{tabular}}
	
	\end{center}
\end{table*}

\begin{remark}
	\label{remark5}
	Considering Table 1, we note that the number of vanishing flats does not contain information of $\omega_0$ and $\omega_2$, and hence remains the same for the case of the ninth class of power functions ($d=2^{\frac{n+3}{2}}-1$, $n \geq7, \ n$ odd), though the differential spectra and the values of ${\cal{D}}(x^d)$ differ depending on $n$ being divisible by $3$ or not, see Remark III.5 in \cite{vanishing}.	
\end{remark}

Now we consider functions with two-valued differential spectrum, i.e., functions $G$ with $\delta_G(a,b)$ taking only two values for all $a \in  \mathbb{F}_{2^n}^*, \  b\in \mathbb{F}_{2^n}. $
In fact, $\delta_G(a,b)\in \{0,2^s\}$ for some integer $s\geq 2$.

\begin{theorem} \label{thm:two-valued}
	Let $G$ be a function having a two-valued differential spectrum. If $G$ is not APN, then
	\begin{align}\label{two-valued}
		\frac{1}{2}(q-1)(3q+2)   \leq      \mathrm{APN\textit{-}def}(G) \leq \frac{1}{4}(q^2+4q+4)(q-1).
	\end{align}
	Moreover, $G$ attains the lower bound if and only if $G$ has differential uniformity $4$, and $G$ attains the upper bound if and only if $G$ is linear.
\end{theorem}
\begin{proof}
	Since %$G$ has a two-valued differential spectrum, 
	%for some integer $s\geq 2$, we have 
	$\delta_G(a,b)\in \{0,2^s\}$ for all $a, b\in \mathbb{F}_q$ with $a\neq 0$, $q={2^n}$, and $G$ is not APN by our assumption, we have $s\neq 1$. In particular, $\nabla_G(a,x)=2^s$ for all $a,x\in \mathbb{F}_q$ with $a\neq 0$. Therefore, the parameters in Equation (\ref{sac}) are $t_a=2^{n-s}$ and $k_i=2^{s-1}$ for all $i=1,\ldots, 2^{n-s}$, $a\in \mathbb{F}_q^*$. Moreover, we have $\sum_{a \in \F_q^*}\chi_a=0$. Therefore, 
	\begin{align}\label{D}
		{\cal{D}}(G) &
		%=-\sum_{a \in \F_q^*}\sum_{x \in \F_q} \frac{2}{\delta_F(a, x)-2}{\delta_F(a, x)/2 \choose 2}
		= -(2^n-1)2^{n-s}2^{2s-2}= -(2^n-1)2^{n}2^{s-2}=-(q-1)q 2^{s-2}.
		%&=-(2^n-1)2^{n}2^{s-2}=- q(q-1)2^{s-2}   
	\end{align} 
	Since $2\leq s \leq n $, we have the following bounds by Equation (\ref{D}).
	\begin{align}\label{D2}
		- \frac{(q-1)q^2}{4} \leq {\cal{D}}(G) \leq -\frac{q(q-1)}{2}.
	\end{align}
	Then the bound in (\ref{two-valued}) is obtained from (\ref{D2}). Note that $G$ attains the upper bound (or the lower bound) in (\ref{D2}) if and only if $s=2$ (or $s=n$), i.e., $G$ has differential uniformity $4$ (or $G$ is linear). 
	$\hfill\square$
\end{proof}

\begin{remark}
The upper bound in (\ref{D2}) shows that non-APN functions with two-valued differential spectrum are not quasi-APN with respect to $\mathrm{APN\textit{-}defect}$, even when such a function is differentially $4$-uniform.  
\end{remark}

Now we investigate the Dembowski-Ostrom (DO) type polynomials, i.e., the polynomials of the form $G(x)=\sum_{1\leq j<i<n}c_{ij}x^{2^i+2^j}$. We recall that for any $a\in \mathbb{F}_q$,
\begin{align*}%\label{eq:DFa}
	D_aG(x)+G(a)=G(x+a)+G(x)+G(a)=\sum_{1\leq j<i<n}c_{ij}(a^{2^j}x^{2^i}+a^{2^i}x^{2^j}).
\end{align*}
That is, if $G$ is a DO polynomial then $D_aG(x)+G(a)$ is a linear function.
\begin{theorem}\label{DOpol}
	Let $G$ be a Dembowski-Ostrom (DO) type polynomial, and let $s_a$ be the dimension of the kernel of $D_aG(x)+G(a)$. Set
	$N_G=\{s_a>1\; : \; a\in \mathbb{F}_q, \; a\neq 0 \}$.
	Then we have
	\begin{align}\label{DO}
		\mathrm{APN\textit{-}def}(G)= (q+1)|N_G|+ q\sum_{s_a\in N_G}  2^{s_a-2} .
	\end{align}
\end{theorem}
\begin{proof}
	Note that $D_aG(x)+G(a)$ is a $2^{s_a}$-to-$1$ map as $s_a$ is the dimension of the kernel of $D_aG(x)+G(a)$. 
	Hence the number of $a$'s for which $D_aG(x)+G(a)$ is a $2$-to-$1$ map is $q-1-|N_G|$, i.e., $\sum_{a \in \F_q^*}\chi_a=q-1-|N_G|$. We again consider the expression given in (\ref{sac}). For a  given $s_a\in N_G$, the parameters in Equation (\ref{sac}) are $t_a=2^{n-s_a}$ and $k_i=2^{s_a-1}$ for all $i=1, \ldots,t_a$, since $q=2^n$.
	Therefore, we have
	\begin{align*}
		{\cal{D}}(G) & =((2^n-1)-|N_G|)(2^n+1)- \sum_{s_a\in N_G} 2^{n-s_a} 2^{2s_a-2}\\
		& =((2^n-1)-|N_G|)(2^n+1)- 2^{n}\sum_{s_a\in N_G}  2^{s_a-2}\\
		%  & =((q-1)-|N_G|)(q+1)- q\sum_{s_a\in N_G}  2^{s_a-2} \\
		& = (q^2-1)-(q+1)|N_G|- q\sum_{s_a\in N_G}  2^{s_a-2}
	\end{align*}
	and obtain Equation (\ref{DO}).
		$\hfill\square$
\end{proof}
\begin{example} Let $G(x)=x^{2^t+1}$ be the Gold function. Then 
	$D_1G(x)+G(1)=G(x+1)+G(x)+G(1)=x^{2^t}+x$. Let $\mathrm{gcd}(n,t)=s>1$, $q={2^n}$. $G$ being a power function, we have $s_a=s$ for all $a\in \mathbb{F}_q^*$. In particular, we have $|N_G|=q-1$. Then Theorem \ref{DOpol} implies that
	\begin{align*}
		\mathrm{APN\textit{-}def}(G)=q^2-1+q(q-1) 2^{s-2}, \quad \text{i.e.,} \quad   {\cal{D}}(G)=- q(q-1) 2^{s-2} .
	\end{align*}
\end
{example}

\begin{remark}
	By our definition, a non-APN  quadratic power function $G$ is not quasi-APN with respect to $\mathrm{APN\textit{-}defect}$, since ${\cal{D}}(G)<0$. This is expected, of course, because of the lack of 2-to-1 derivatives. 
\end{remark}

\section{APN-defect and modifications of the inverse function}

		We have noted above, see Example \ref{inversevf} and Remark \ref{reminv}, that the inverse function $F(x)=x^{2^n-2}$ over $\F_{2^n}$ for even $n$ has the property that $\nabla(a,x)=2$, for any $ x \notin D_F(a,0)$, and $\nabla(a,0)=4$, $a \in \F_{2^n}^*$. In other words, it has the maximum possible number of 2-to-1 derivatives among  non-APN power functions, attaining the lower bound in (\ref{defect-mono}) for the APN-defect. %among non-APN power functions, 
		It also is the only known differentially 4-uniform power function so far, where the associated code $C_F$ has no codewords of weight 4, see Problem 1 in \cite{weight3}. \\

$F$ is known to have the best known nonlinearity $2^{n-1}-2^{n/2}$, and maximum algebraic degree $n-1$. Moreover, when modified at a small number of elements of $\F_{2^n}$, or on particular sets, it yields permutations with high nonlinearity and algebraic degree, see for instance \cite{JKKson,37}. 	
	Although $F$ itself is weak against algebraic attacks (putting $y=x^{-1}$ one gets $x^2y=x$),  its modifications that are obtained by composing it with cycles are not, see Theorem \ref{ccds} and Example \ref{Fab}, below.\\ 
	
	As a consequence, the differential behaviour of modifications of $F$ are studied widely. In general, one composes it with some permutations and it is convenient to view such  modifications as permutations of particular Carlitz ranks. We now explain how the concept of Carlitz rank and modifications of the inverse function are linked. 

\subsection{Modifications of the inverse permutation: Carlitz rank }

		We recall that the inverse function is a building block of all permutation polynomials of $\F_{q}$ in the following sense.
	Consider the group $P_q$ of permutation polynomials over $\F_q$ of degree less than $q$, under the operation of  composition and subsequent reduction modulo $x^q-x$. A well-known result of Carlitz \cite{carlitz} determines a set of generators of the group $P_q$  to be linear polynomials $ax+b$ for $a,b \in \F_q, ~a \neq 0$, and $x^{q-2}.$\\
	
	We now consider polynomials $F_k \in \F_q[x]$ that are defined recursively as
	$ F_k(x)= F_{k-1}(x)^{q-2}+a_{k+1}$,  $k\geq 1$ and 
	$F_0(x)=a_0x+a_1$, where $a_0\neq 0 $. \\ 
	
	An immediate consequence of the above mentioned result of Carlitz is that any permutation $\sigma$  of $\F_q$ can be represented by a polynomial of the form $F_k$ for some $k\geq 0$, 
	i.e.,
	there is a polynomial
	\begin{equation}
	\label{fn} % can be represented by
	F_k(x)=(\cdots ((a_0x+a_1)^{q-2}+a_2)^{q-2}+ \cdots +a_k)^{q-2}+a_{k+1},
	\end{equation}
	satisfying $\sigma(c)= F_k(c)$ for all $c \in \F_q$.\\ 
	
	The polynomial $F_k \in \F_q[x]$ in (\ref{fn}) that represents $\sigma$  is not unique, however one can consider the least number of the monomials $x^{q-2}, $ needed to obtain $\sigma$. The authors of \cite{aksoy} call this number the \textit{Carlitz rank} of 
	the permutation $\sigma$ (or the permutation polynomial $F_k$), and denote it by $Crk (\sigma)=Crk(F_k)$.
	This concept has been in use for over a decade, in relation to diverse problems concerning, for instance, pseudorandom number generation \cite{adt,meidl}, uniform distribution theory \cite{pausinger2,pausinger}, theory of polynomials over finite fields \cite{difference}, and cryptography \cite{leyla,decomp}. 	
	 We refer the reader to	
	 the surveys \cite{anbar2,topuzoglu} for details. The relation to  differential uniformity was first pointed out in \cite{Ayca}. The paper \cite{ccds} poses some open problems of current interest, including those on the differential behaviour of the modifications of $F(x)=x^{q-2}$. The concept of Carlitz rank also appeared in some recent work concerning constructions of differentially 4-uniform permutations with additional properties, see \cite{JK22,boomerang}. \\

		We recall that any permutation of $\F_q$ can be expressed as a composition of a permutation $\tau$ and the inverse permutation. The following result, proved in \cite{aksoy}, is a special case of the version given as Theorem 25 in \cite{ccds}, and demonstrates how the concept of Carlitz rank is linked to the study of modifications of the inverse function. We use essentially the same notation as in \cite{ccds}.
	For a permutation $\tau$ of $\F_q$, $\rm{supp}(\tau)$ denotes the set of all elements of $\F_q$, which are moved by $\tau$.\\ 
	\begin{theorem}\label{ccds}
		\cite{ccds}
		Let $H$ be a permutation polynomial of $\F_q$ and $\sigma_H$ be the permutation induced by $H$. Suppose $\sigma_H$ has the cycle decomposition
		\begin{equation}
		\label{cycle}
		%f(x)=\tau_{1}\ldots \tau_{m}\bar{R_{s}}(x),
		\sigma_H=\tau_{1}\cdots \tau_{m} \cdot \sigma_F,
		\end{equation}
		where $\tau_{1},\ldots, \tau_{m}$ are disjoint cycles of length $l(\tau_{j})=l_j \ge 2$, $1 \le j \le m$,  
		and $\sigma_F$ is the permutation induced by %the monomial
		 $F(x)=x^{q-2}$.
		Put $\tau=\tau_{1}\cdots \tau_{m}.$ 
		Then
		there exists $F_k \in \F_q[x]$ of the form (\ref{fn}), with $\sigma_H (c)=F_k(c)$ for all $c \in \F_q,$ 
		where
		\begin{enumerate}
			\item[(i)] $k=m+\sum_{j=1}^{m}l_{j}-1$ if $0 \in \rm{supp}(\tau);$
					\item[(ii)] $k=m+\sum_{j=1}^{m}l_{j}+1$ if 
			$0 \notin \rm{supp}(\tau).$	
							\end{enumerate}
		In both cases, $Crk(H) = k$ if $k < (q-1)/2$.	
	\end{theorem} 
	\begin{example}
	Let $\alpha, \beta,\delta,\gamma$ be distinct, non-zero elements in a finite field $\F_q$ with $q\geq 16$. Suppose $\tau$ in Theorem \ref{ccds} above is $\tau=(1/\alpha \  1/\beta)(1/\delta \ 1/\gamma)$. Then, $\sigma_H$ in (\ref{cycle}) is the permutation given by 
	$\sigma_H(\alpha)=1/\beta$, $\sigma_H(\beta)=1/\alpha$, $\sigma_H(\delta)=1/\gamma$, $\sigma_H(\gamma)=1/\delta$, and $\sigma_H(c)=c^{q-2}$ for $c \in \F_q \setminus \{\alpha, \beta,\delta,\gamma\}$. It follows by {\it{(ii)}} above that $Crk(H) = 7$.
	\end{example}
	
	\begin{remark}
		One can actually use an algorithm given in \cite{aksoy} to find coefficients of $F_k$ in Theorem \ref{ccds}, as the following example illustrates. 
	\end{remark}
	
	\begin{example}
		\label{Fab}
		Let $\alpha \neq \beta$ be elements of $\F_q^*$ with $q=2^n, n\geq 4$. Suppose $\tau=(1/\alpha \ 1/\beta)$. Then, $\sigma_H$  is simply the permutation which differs from the inverse permutation only at $\alpha, \beta$, where the inverses of $\alpha$ and $\beta$ are interchanged. By   
		 {\it{(ii)}} above, we have $Crk(H) = 4$. The algorithm given in   \cite[page 435]{aksoy} yields the polynomial $F_4$ to be
$$F_4(x)=\left(\left(\left(\left(\frac{({\alpha}^2+{\beta}^2)x}{{\alpha}^2{\beta}^2}\right)^{q-2}+\frac{{\alpha} {\beta}^2}{{\alpha}^2+{\beta}^2}\right)^{q-2}+
\frac{{\alpha}+{\beta}}{{\alpha}{\beta}}\right)^{q-2}+\frac{{\alpha}{\beta}}{{\alpha}+{\beta}}\right)^{q-2}+\frac{1}{{\beta}}. 
$$		 
	
	\end{example}
	
	In this work, we focus on permutations $H$, where $\tau$ in Theorem \ref{ccds} is a single cycle. This first part deals with a transposition only. In the forthcoming second part, we also consider modifications of a permutation $G$, which is not necessarily the inverse function. If the values of $G$ at pairwise distinct $\ell$ elements, say at $\alpha_1, \ldots, \alpha_{\ell}$, are interchanged so as to obtain the permutation $\sigma_H$ satisfying
	%\begin{equation}%\label{sigma}
$$	\sigma_H=(\sigma_G(\alpha_1)\ldots ~\sigma_G(\alpha_{\ell}))\cdot \sigma_G,$$
	%\end{equation}
	we denote the permutation polynomial that induces $\sigma_H$ by $G_{\alpha_1, \ldots, \alpha_{\ell}}$, i.e., we put $H=G_{\alpha_1, \ldots, \alpha_{\ell}}$.\\

	In particular, we have
	\begin{equation}
	\label{eqfab}
	\sigma_{F_{\alpha_1, \ldots, \alpha_{\ell}}}=((\alpha_1)^{q-2}\ldots (\alpha_{\ell})^{q-2})\cdot \sigma_F,
	\end{equation}
where $F(x)=x^{q-2}.$ We sometimes abuse the notation and use  
$F_{\alpha_1, \ldots, \alpha_{\ell}}$ to denote the permutation of $\F_q$ that is induced by $F_{\alpha_1, \ldots, \alpha_{\ell}}$.\\ 

 Obviously, we have $Crk(F_{\alpha_1, \ldots, \alpha_{\ell}})=\ell+2$ if $\alpha_1 \cdots \alpha_{\ell} \neq 0$, $\ell < (q-5)/2$, and $Crk(F_{\alpha_1, \ldots, \alpha_{\ell}})=\ell$, if $\alpha_1 \cdots \alpha_{\ell} = 0$, $\ell < (q-1)/2$.
 	
	\begin{example} 
	With the notation of (\ref{eqfab}), the polynomial $H$ in Example \ref{Fab} is
	$F_{\alpha,\beta}$ and we have $F_{\alpha,\beta}(c)=F_4(c)$ for every $c\in \F_q$.
	\end{example}

	\begin{remark}
One of the open questions posed in \cite{ccds} was to find the smallest $k$ (and the smallest $q>5$) such that there exists an APN permutation $G \in \F_q[x]$ with $Crk(G)=k$. This problem was partially solved in \cite{boomerang}, where the authors showed that when $q=2^n$ and $n$ is even, 
the Carlitz rank of an APN permutation 
$G:\F_{2^n} \rightarrow \F_{2^n}$ must be at least $\frac{2^{n-1}-1}{3}$. Hence one cannot expect to find APN permutations with small Carlitz rank (when compared to the field size).  
	\end{remark}	
	
	We are now ready to investigate the APN-defect of modifications of the inverse function with small Carlitz rank. In this first part of our work we cover the simplest case. Other cases will be presented in the forthcoming second part. 	
	
	\subsection{APN-defect of $F_{0,\alpha}$}

The permutation of $ \F_{2^n}$ which is induced by the function $F_{0,\alpha}$, $\alpha\neq0$, is defined as in (\ref{eqfab}). In this simple case, we have
\begin{equation}
\label{foalfa}
	\sigma_{F_{0,\alpha}}=(0, \alpha^{{2^n}-2})\cdot \sigma_F,
%	\sigma_{F_{0,\alpha}}=(0, \alpha^{2^n-2})\cdot \sigma_F,	
\end{equation}
 in other words, $\sigma_{F_{0,\alpha}}(0)=\alpha^{-1}$, $\sigma_{F_{0,\alpha}}(\alpha)=0$, and $\sigma_{F_{0,\alpha}}(c)=c^{-1}$ for all $c\in \F_{{2^n}}\setminus \{0, \alpha \}$. If $n\geq 3$, $Crk(F_{0,\alpha})=2.$ Again, the algorithm given in   \cite[page 435]{aksoy} yields the polynomial 
 $$F_2(x)=(({\delta}^2x+{\delta})^{{2^n}-2}+{\delta}^{-1})^{{2^n}-2}
 +\delta, ~~~~ \delta=1/\alpha,
 $$
 which satisfies $F_2(c)=F_{0, \alpha}(c)$ for every $c\in \F_{2^n}$ when $n\geq3.$
 \\ 
 
In order to calculate the APN-defect of $F_{0,\alpha}$, and the partial quadruple system associated to it, we first present a detailed description of its differential behaviour. \\

The following theorem determines all entries in the difference  square for $F_{0, \alpha}$, hence also the differential spectrum of $F_{0, \alpha}$.  
One needs to describe how the derivatives $D_aF_{0, \alpha}$, which are changed after the modification, are related to each other and to the remaining entries (of the difference square for the inverse function $F$). This naturally leads to the trivial problem of checking conditions for solvability of some particular quadratic equations. Hence, the next theorem is straightforward to obtain though it may look rather technical. \\ 

Since $F_{0, \alpha}$ was studied extensively in search for functions of differential uniformity 4 with additional favourable properties, the proof of Theorem \ref{f0a} can essentially be found in the literature, see for instance \cite{25,27,42}. However, our proof is somewhat simpler and more specifics are included here in order to obtain Corollary \ref{spectf0alpha}, Theorems \ref{apndeff0a} and \ref{pqsf0a}.
We therefore present details of the proof for the convenience of the reader, and also to fix the terminology for later reference. Remark \ref{expfoalpha} below describes some of the earlier work on $F_{0, \alpha}$.\\

In what follows, we denote by $\T(z)$ the absolute trace of $z \in \F_{2^n}$, i.e., $\T(z)=\T_{\F_{2^n}/\F_2}(z)=z+z^2+z^{2^2}+\cdots+z^{2^{n-1}}$.
For simplicity, we write
$$Q(a,b,c;x)=ax^2+bx+c  $$ or just $Q(a,b,c)$, when the variable $x$ does not need to be specified. Recall that when $ab \neq 0$, and $\T(ac/b^2)=0$, the polynomial 
$Q(a,b,c)$ has $2$ distinct roots in $\F_{2^n}$. We put $$R_{Q(a,b,c)}=\{\rho \in \F_{2^n}:
Q(a,b,c;\rho)=0    \}.$$ 
Recall that when $\rho \in R_{Q(1,1,c)}$, then $\rho+1 \in R_{Q(1,1,c)}$, and 
$a\rho \in R_{Q(1,a,a^2c)}$ for $a \in \F_{2^n}^*$.\\

The reader may find it useful at this point to look closely at the difference square in Example 1. Indeed, the function $G$ in Example 1 is $F_{0, \alpha}$, where $\alpha=\zeta$.

\begin{theorem}\label{f0a}
	Let $\alpha \in \F_{2^n}^*$ be arbitrary,  $F(x)= x^{{2^n}-2}$, and the permutation polynomial 
	$F_{0, \alpha}$ be as defined in (\ref{foalfa}) above. %Put $\gamma=\frac{\alpha}{a + \alpha}$ for $a \in \F_{2^n}^*$. 
	\\
	
	\noindent
	Suppose that $n$ is odd. 
	\begin{itemize}
		\item[I.]
	$\chi_a=1$ if and only $a \in \F_{2^n}^*$ satisfies
		
		I.i)
		$a=\alpha$ or
		
		I.ii)  $\T(\frac{\alpha}{a + \alpha})(\T(\frac{\alpha}{a})+1)=1$ .\\ 
		%$(\T(\gamma)(\T(\gamma^{-1})=1$ .\\ 
		\item[II.]
		$\nabla(a,x)=4$ if and only if $a\neq 0, \alpha$ and one of the following holds. 
		
		II.i)  $\T(\frac{\alpha}{a + \alpha})=0$ and $x\in D_{F_{0, \alpha}}(a,0)= \{0, \ a, \ a\rho, \ a\rho+a  \}$, where $\rho \in R_{Q(1,1,\frac{\alpha}{a + \alpha})} $.

		II.ii) $\T(\frac{\alpha}{a})=1$ and $x\in D_{F_{0, \alpha}}(a,\alpha)= \{\alpha, \ \alpha+a, \ a\rho, \ a\rho+a   \}$, where $\rho \in R_{Q(1,1,\frac{a + \alpha}{a})} $. \\
	
		\item[III.]
		$\nabla(a,x)=2$ for all values of $x$, which are not in the sets
		$D_{F_{0, \alpha}}(a,0)$ in (II.i) and $D_{F_{0, \alpha}}(a,\alpha)$
		in (II.ii).
		 \\             %$a\neq \alpha$ and one of the following holds. 

		%$\T(\frac{\alpha}{a + \alpha})=1$ or $\T(\frac{\alpha}{a}) =1$.
		
		%$\T(\frac{\alpha}{a}) =0$.
		\item[IV.] There is no $x \in \F_{2^n}$, satisfying $\nabla(a,x)=2$ for all  $a \in \F_{2^n}^*$.\\
				
		\noindent
		Suppose that $n$ is even and $\omega\in \F_4 \setminus \F_2$.\\
		\item[V.]
	$\chi_a=1$ if and only if
		$a \notin \{0, \alpha, \alpha\omega, \alpha \omega^2\}$ and
				$\T(\frac{\alpha}{a + \alpha})=\T(\frac{\alpha}{a})=1$.\\

		\item[VI.]
		$\nabla(a,x)=4$ if and only if one of the following holds.
		
		VI.i) $a\notin \{0, \alpha,\alpha\omega, \alpha \omega^2   \}$, 
		$\T(\frac{\alpha}{a})=0$ and $x \in D_{F_{0, \alpha}}(a,\alpha)=	\{\alpha, \  \alpha +a, \ a\rho, \ a\rho+a \}, $ where  $\rho \in R_{Q(1,1,(a+\alpha)/a)} $.

		VI.ii) $a \notin \{0, \alpha, \alpha\omega, \alpha \omega^2\}$,
		$\T(\frac{\alpha}{a + \alpha})=0$, and $x \in D_{F_{0, \alpha}}(a,0)=\{0, \ a, \ a\rho, \  a\rho+a \}$, where  $ \rho \in R_{Q(1,1,\alpha/(a+\alpha))} $.	\\	
	VI.iii) 
			$a=\alpha$ and  $x \in  D_{F_{0, \alpha}}(\alpha,\alpha)=\{0, \alpha, \alpha\omega, \alpha \omega^2   \}$.\\			
		VI.iv) $n \equiv 2 \  (\hbox{mod}  \ 4)$, $a\in \{\alpha\omega, \alpha \omega^2   \}$, and $x \in D_{F_{0, \alpha}}(a,0)=\{0, a, \alpha, \alpha+a  \}= \{0, \alpha, \alpha\omega, \alpha \omega^2   \}$. \\
		
		\item[VII.]
		$\nabla(a,x)=6$ if and only if $n \equiv 0 \  (\hbox{mod}  \ 4)$,
		$a\in \{\alpha\omega, \alpha \omega^2   \}$, and $x \in D_{F_{0, \alpha}}(\alpha\omega,0)=D_{F_{0, \alpha}}(\alpha\omega,\alpha)=\{0, \alpha, \alpha\omega, \alpha\omega^2, \alpha\omega\rho, \alpha\omega(\rho+1)\}$, with 
		$\rho \in R_{Q(1,1,\omega)} $,
		or $x \in D_{F_{0, \alpha}}(\alpha\omega^2,0)=D_{F_{0, \alpha}}(\alpha\omega^2,\alpha)=\{0, \alpha, \alpha\omega
		, \alpha\omega^2, \alpha\omega^2\rho, \alpha\omega^2(\rho+1)\}$, where $\rho \in R_{Q(1,1,1/\omega)} $. \\

		\item[VIII.]
		$\nabla(a,x)=2$ for all values of $x$, which are not in the sets
		$D_{F_{0, \alpha}}(a,y)$ in (VI)-(VII).
		\\  
		\item[IX.]  $x \in \F_{2^n}$ satisfies $\nabla(a,x)=2$ for all  $a \in \F_{2^n}^*$ if and only if $x \in  \F_{2^n} \setminus \{0,  \alpha,  \alpha \omega, \alpha \omega^2\}$, 
			 and  $\T(\frac{\alpha}{x + \alpha})=1$.\\

	\end{itemize}	
\end{theorem}

\begin{proof}
	We consider  $D_aF_{0, \alpha}(x)=F_{0, \alpha}(x) + F_{0, \alpha}(a+x)$, where $a \in \F_q^*$.
	%The values of $D_aF_{0, \alpha}(x)$ can obviously be determined as follows.
	Assuming  $a \neq \alpha$ one has, 
	\begin{itemize}
		\item [i)]  
		%i) 
		$D_aF_{0, \alpha}(0)=
		%\frac{1}{\alpha} + \frac{1}{a} = 
		\frac{a + \alpha}{a \alpha}$, and 
		
				\item [ii)] 
				$ D_aF_{0, \alpha}(\alpha)= 
		\frac{1}{a +  \alpha}$. 
	\end{itemize}       
	\noindent In case $a=  \alpha$ we have,
	\begin{itemize}
		\item [iii)]     
		%	iii)        
		$D_{\alpha}F_{0, \alpha}(0)=
		D_{\alpha}F_{0, \alpha}(\alpha)=%F(0) + F(\alpha)= \frac{1}{ \alpha}  + 0 = 
		\frac{1}{  \alpha}$. % when $x=0$ or $x=a$. 
	\end{itemize}       
	\noindent When $a \in \F_{2^n}^*$ is arbitrary and  $x\neq0, \alpha, a, a + \alpha$, 
	\begin{itemize}
		\item [iv)] 
		%iv)      
		$ D_aF_{0, \alpha}(x)=%F(x) + F( x + a )=   \frac{1}{x}  + \frac{1}{a + x} = 
		\frac{a}{ax + x^2}$.
		
	\end{itemize}
	In order to determine how the derivatives given in (i)-(iv) are related to each other we need to look for solutions of the quadratic equations below.\\ 
	
	The values of the derivatives in (i) and (ii) are the same, i.e.,
	$\frac{a + \alpha}{a \alpha}= \frac{1}{a + \alpha}$ holds for some $a
	$, exactly when 
	\begin{equation}\label{3}
	a^2 + \alpha a + \alpha^2 =0 
	\end{equation}
	has solutions in $\F_{2^n}^*$.    
	%The Equation (3) has two solutions in $\F_{2^n}^*$ 
	Hence $D_{F_{0, \alpha}}(a, 0)= D_{F_{0, \alpha}}(a, \alpha )$ %, implying
	 %$\nabla_{F_{0, \alpha}}(a, 0)=\nabla_{F_{0, \alpha}} (a, \alpha )$
	if and only if $\T(1)=0.$ \\

	Similarly, $\frac{a + \alpha}{ a \alpha}=\frac{a}{ax + x^2}, a \neq \alpha$, or 
	\begin{align}\label{4} 
	(a + \alpha)x^2 + (\alpha a + a^2)x + \alpha a^2=0
	\end{align}
	has solutions in  
	$\F_{2^n}^*$ if and only if 
	$\T\left(\frac{\alpha}{ a + \alpha}\right)=0$. When this trace condition is satisfied, one has  $\nabla (a, 0) \geq 4 $ and  %for $x \in \{a\rho: \rho \in R_{Q(1,1,(\alpha a^2)/(a+\alpha))} \}$, so that    $	
$	D_{F_{0, \alpha}}(a, 0) \supseteq \{0,a,a\rho, a\rho+a \}$, where  $ \rho \in R_{Q(1,1,\alpha/(a+\alpha))}  $.\\
	
	%Clearly, $\frac{1}{ a + \alpha} =\frac{1}{\alpha}$ never holds as $a \in \F_{q}^*$.\\
	The derivatives in (ii) and (iv) with $a \neq \alpha$,
	%are the same, i.e., $ \frac{1}{ a + \alpha}=\frac{a}{ax + x^2}$, or\\
	lead to the equation 
	\begin{align}\label{5}
	x^2 + ax+ \alpha a + a^2 =0,
	\end{align}
			and the condition
	$ %\T\left(\frac{\alpha + a}{ a}\right)=
	\T(\frac{\alpha}{ a})+\T(1)=0$. The values of $a$ satisfying this condition yield $\nabla(a, \alpha) \geq 4 $ and 	
	$D_{F_{0, \alpha}}(a, \alpha) \supseteq \{\alpha, \alpha +a,a\rho, a\rho+a \}, $ where  $\rho \in R_{Q(1,1,(a+\alpha)/a)} $.\\
	%\end{align*}
	
	Finally, the case $a=\alpha$, i.e.,  $\frac{1}{\alpha} =\frac{\alpha}{\alpha x + x^2}$ leads to the equation
	\begin{align}\label{6}
	x^2 + \alpha x +\alpha^2 =0,   
	\end{align}
	%which has solutions in $\F_{2^n}$ if and only if 
	and the condition
	$\T(1)=0$. When this condition is satisfied, one has  $\nabla(\alpha, 0)=4 $ and 	
	$D_{F_{0, \alpha}}(\alpha, 0)=\{0, \alpha ,\alpha\omega, \alpha\omega^2 \}. $ \\	
	
	Now suppose that $n$ is odd. Since $\T(1)=1$, Equations (\ref{3}) and (\ref{6}) have no solutions in $\F_{2^n}$. Therefore the only values of $a, x$ with 
	$\nabla_aF_{0, \alpha}(x)=4$ are those given in part \textit{(II)} above. \\
	
	Note that $\chi_{\alpha}=1$ since Equation (\ref{6}) does not have a solution in $\F_{2^n}$, proving part \textit{(I.i)}.\\
	
	When $\T(\frac{\alpha}{a + \alpha})=1$ and $\T(\frac{\alpha}{a})=0$, Equations (\ref{4}) and (\ref{5}) have no solutions, which imply that for values of $a$ satisfying conditions in \textit{(I.ii)}, there is no $x$ with $\nabla(a,x)>2$.\\
	
	Part \textit{(III)} holds since we have already considered all the cases where $\nabla  (a,x) >2$. \\
	
	In order to prove part \textit{(IV)}, we need to consider the same equations as above, but look for the solutions of (\ref{4})-(\ref{6}) in $a$, rather than in $x$. For instance, Equation (\ref{4}) can be expressed as  
%For instance, the Equation \eqref{4} above can be exressed as
\begin{equation}\label{10}
(x + \alpha)a^2 + (\alpha x + x^2)a + \alpha x^2=0,
\end{equation}
$x \neq \alpha$, which holds for some 
$a \in\F_{2^n}^*$ if and only if 
%$\T\left(\frac{(x + \alpha) \alpha x^2}{ x^2(x + \alpha)^2}\right)=$
$\T\left(\frac{\alpha}{ x + \alpha}\right)=0.$\\

Similarly,
\begin{align} \label{11}
a^2 + (x + \alpha)a + x^2=0
\end{align}
has solutions $a \in\F_{2^n}^*$ if and only if 
$ \T\left( \frac{x^2}{ x^2 + \alpha^2}\right)=$ 
$\T(1) + \T\left(\frac{\alpha}{ x + \alpha}\right)=0.$\\

When $n$ is odd, $\T(1)=1$ implies that either % $\T\left(\frac{\alpha}{ x + \alpha}\right)=0$, then  Equation
 \eqref{10} or 
%has solutions, otherwise Equation  
\eqref{11} 
has solutions in $\F_{2^n}$, i.e., whatever the value of $x\neq \alpha$ is, there is an $a \in \F_{2^n}^*$, such that $\nabla (a,x) > 2$. Note also that  $\nabla (\alpha, \alpha) =4$, since $\alpha \in D_{F_{0, \alpha}}(\alpha, 0)=\{0, \alpha, \alpha\omega, \alpha\omega^2  \}$. Hence,  part \textit{(IV)} follows. \\

	Suppose now that $n$ is even. As in the case of odd $n$, it is easy to see that parts \textit{(V)}, \textit{(VI.i)} and \textit{(VI.ii)} hold. \\
	
	For proving parts \textit{(VI.iii)} and \textit{(VII)}, i.e., when $a\in \{\alpha\omega, \alpha \omega^2   \}$, we recall that 
	$\T(\omega)=0$, when  $n \equiv 0 \  (\hbox{mod}  \ 4)$ and $\T(\omega)=1$, when $n \equiv 2 \  (\hbox{mod}  \ 4)$. This immediately follows from $<\omega>=\F_4^*$. Clearly,	$\T(\frac{\alpha}{\alpha\omega})=	\T(\frac{\alpha}{\alpha\omega^2})=\T(\omega)$, and 
	$\T(\frac{\alpha}{\alpha + \alpha\omega})= \T(\frac{\alpha}{\alpha + \alpha\omega^2})= \T(\omega).$\\
	
	Therefore, if $n \equiv 2 \  (\hbox{mod}  \ 4)$, then 
	$\T(\frac{\alpha}{a})=	\T(\frac{\alpha}{a + \alpha})=1$, and Equations (\ref{4}) and (\ref{5}) have no solutions but (\ref{3}) does, implying that   $\nabla  (a,0)=4 $, and $D_{F_{0, \alpha}}(a,0)=D_{F_{0, \alpha}}(a,\alpha)=\{0, \alpha, \alpha\omega
	, \alpha\omega^2  \}$. This proves part \textit{(VI.iii)}.\\ 	
	
	If $n \equiv 0 \  (\hbox{mod}  \ 4)$, then
	$\T(\frac{\alpha}{a})=	\T(\frac{\alpha}{a + \alpha})=0$, and Equations (\ref{4}) and (\ref{5}), as well as Equation (\ref{3}) have  solutions in $\F_{2^n}$. Recall that $a\in \{\alpha\omega, \alpha \omega^2   \}$. Depending on the value of $a$, we get two possibilities for 
	the sets	
	$D_{F_{0, \alpha}}(a,0)=D_{F_{0, \alpha}}(a,\alpha)$. Indeed, we have 
$	D_{F_{0, \alpha}}(\alpha\omega,0)=D_{F_{0, \alpha}}(\alpha\omega,\alpha)$ $=\{0, \alpha, \alpha\omega, \alpha\omega^2,$ $ \alpha\omega\rho, \alpha\omega(\rho+1)\}$, where 
	$\rho \in R_{Q(1,1,\omega)} $, or 
$D_{F_{0, \alpha}}(\alpha\omega^2,0)=D_{F_{0, \alpha}}(\alpha\omega^2,\alpha)=\{0, \alpha, \alpha\omega
, \alpha\omega^2, \alpha\omega^2\rho, \alpha\omega^2(\rho+1)\}$, where $\rho \in R_{Q(1,1,1/\omega)} $. This proves part \textit{(VII)}.\\

Part \textit{(VIII)} holds since all the cases where $\nabla  (a,x) \geq 4 $ are covered above. \\

In order to prove part \textit{(IX)}, we first observe from parts 	
\textit{(VI.iii)} and \textit{(VII)} that when $x \in \{0, \alpha, \alpha\omega
, \alpha\omega^2  \}$, we have $\nabla  (\alpha\omega,x) \geq 4 $ and $\nabla  (\alpha\omega^2,x) \geq 4 $. Note also that both of the equations (\ref{10}) and (\ref{11}) have solutions in $\F_{2^n}$ if and only if 	
$\T(\frac{\alpha}{x + \alpha})=0$. Therefore the assertion follows.  	 
	$\hfill\square$	
	\end{proof}

The parts \textit{(I)}, \textit{(IV)}, \textit{(V)}, and \textit{(IX)} above can be rephrased as follows.  

\begin{corollary}
	\label{spectf0alpha}
		Let $\alpha \in \F_{q}^*$ be arbitrary, $q={2^n}$, $F(x)= x^{q-2}$, and the permutation polynomial 
		$F_{0, \alpha}$ be as defined in (\ref{foalfa}) above. Put $\omega\in \F_4 \setminus \F_2$. Then the following hold. 
		\begin{itemize}
			\item[I.] The row spectrum  of $ F_{0, \alpha}$ is the set  $$R\mathrm{\textit{-}Spec_{q}}(F_{0, \alpha})= \{\alpha\} \cup
			 \left\{a \in \F_{q}\setminus  \{0,\alpha\}:\T\left(\frac{\alpha}{a + \alpha}\right) \left(\T\left(\frac{\alpha}{a}\right)+1 \right)=1  \right\},  $$ 
			 when $n$ is odd. 
			\item[II.]The row spectrum  of $ F_{0, \alpha}$ is the set  $$R\mathrm{\textit{-}Spec_{q}}(F_{0, \alpha})= 
			\left\{a \in \F_{q} \setminus  \{0,\alpha, \alpha\omega, \alpha \omega^2\}: \T\left(\frac{\alpha}{a + \alpha}\right) =\T\left(\frac{\alpha}{a}\right)=1  \right\},  $$ 
			when $n$ is even. 			
			
				\item[III.]
				The column spectrum  of $ F_{0, \alpha}$ is the empty set %$$C\mathrm{\textit{-}Spec_{q}}(F_{0, \alpha})= \emptyset  $$ 
				when $n$ is odd.\\
					\item[IV.]
					The column spectrum  of $ F_{0, \alpha}$ is the set   $$C\mathrm{\textit{-}Spec_{q}}(F_{0, \alpha})= \left	\{x \in \F_{q} \setminus  \{0,\alpha, \alpha\omega, \alpha \omega^2\}: \T(\frac{\alpha}{x + \alpha})=1 \right \},  $$ 
					when $n$ is even.
					\end{itemize}
\end{corollary}

\begin{remark}
	\label{expfoalpha}
	We note that the parts \textit{(III)}, and \textit{(IV)} of the above corollary is known, see \cite[Theorem 2]{swap}. However, our view point does not only simplify and shorten the proof in \cite{swap} considerably, but it also enables us to extend the results on the column spectra to that of other functions easily, for instance, to permutation polynomials of higher Carlitz ranks, see the forthcoming second part of this work.  
\end{remark}

We are now ready to determine the APN-defect of $F_{0,\alpha}$.

\begin{theorem}
	\label{apndeff0a}
		Let $\alpha \in \F_{q}^*$ be arbitrary,  $F(x)= x^{q-2}$, $q=2^n$, and the permutation polynomial 
	$F_{0, \alpha}$ be as defined in (\ref{foalfa}) above.
	The value of the APN-defect of $F_{0,\alpha}$ is given as follows.	

	\begin{itemize}		
\item[I.] Suppose n is odd. Let $|R\mathrm{\textit{-}Spec_{q}}(F_{0,\alpha})|=k$, and $|\{a \in \F_{q}\setminus\{0,\alpha\}: \T(\frac{\alpha}{a + \alpha})=0, \  \mathrm{and} \  \T(\frac{\alpha}{a})=1  \}|=\ell$.  Then, \\
	$$\mathrm{APN\textit{-}def}(F_{0,\alpha})=9q+8\ell-9k-9.$$

\item[II.] Suppose n is even. 
Let $|R\mathrm{\textit{-}Spec_{q}}(F_{0,\alpha})|=k$, and 
$|\{a \in \F_{q}\setminus \{0,\alpha, \alpha\omega, \alpha \omega^2\}: \T(\frac{\alpha}{a + \alpha})=\T(\frac{\alpha}{a})=0  \}|=s$, where $\omega\in \F_4 \setminus \F_2$.
%$\T(\frac{\alpha}{a + \alpha})=\T(\frac{\alpha}{a})=0$
 Then,
		$$\mathrm{APN\textit{-}def}(F_{0,\alpha})=9q + 8s - 9k - 9, ~\mathrm{when}~  n \equiv 2 \  (\hbox{mod}  \ 4) ,$$ and 
		
		$$\mathrm{APN\textit{-}def}(F_{0, \alpha})= 9q + 8s - 9k + 5,  ~\mathrm{when}~ n \equiv 0\  (\hbox{mod}  \ 4) .$$
	
	\end{itemize}	
	
	\noindent

\end{theorem}

\begin{proof}

	Recall the Equations (\ref{sa}), (\ref{sac}), (\ref{weightedsac}), (\ref{dg}) above;
	
	$$ {\cal{D}}(G)=\sum_{a\in \F_{q}^*}(|S_a|-w_{S_a^c}+\chi_a).  $$
\begin{itemize}		
	\item[I.]	
	Using Theorem \ref{f0a}, parts \textit{(I), (II)}, we have $$\sum_{a\in \F_{q}^*}|S_a|= kq+\ell(q-8)+(q-k-\ell-1)(q-4)=q^2-5q-4\ell+4k+4,$$
$$\sum_{a\in \F_{q}^*}w_{S_a^c}=(2^2+2^2)\ell+2^2(q-k-\ell-1)=4q+4\ell-4k-4,$$ and 
$\sum_{a\in \F_{q}^*}\chi_a=k$. Hence, \\

${\cal{D}}(F_{0, \alpha})=q^2-9q-8\ell+9k+8 $, and the assertion follows. \\

      \item[II.] In the case  $ n \equiv 2\  (\hbox{mod}  \ 4) $, the values of $\sum_{a\in \F_{q}^*}|S_a|$ and 
$\sum_{a\in \F_{q}^*}w_{S_a^c}$ can be calculated as above to get
$$\sum_{a\in \F_{q}^*}|S_a|= q^2-5q-4s+4k+4,$$
$$\sum_{a\in \F_{q}^*}w_{S_a^c}=4q+4s-4k-4.$$

When $ n \equiv 0\  (\hbox{mod}  \ 4) $, parts \textit{(V), (VI), (VII)} of  Theorem \ref{f0a} imply
$$\sum_{a\in \F_{q}^*}|S_a|= kq+s(q-8)+(q-k-s-3)(q-4)+2(q-6)=q^2-5q-4s+4k,$$
$$\sum_{a\in \F_{q}^*}w_{S_a^c}=8s+4(q-k-s-3)+2 \cdot 9=4q+4s-4k+6,$$ and 
$\sum_{a\in \F_{q}^*}\chi_a=k$. 
Therefore,\\ 
	
${\cal{D}}(F_{0, \alpha})=q^2-9q-8s+9k+8 $, when  $ n \equiv 2\  (\hbox{mod}  \ 4) $, and\\ 

${\cal{D}}(F_{0, \alpha})= q^2-9q-8s+9k-6$,  when  $ n \equiv 0\  (\hbox{mod}  \ 4) $.
\end{itemize}	
$\hfill\square$	
\end{proof}

\begin{remark}
	Obviously, the quantities $k,l,s$ mentioned in Theorem \ref{apndeff0a}  are independent of the choice of $\alpha \in \F_{q}^*$. In other words, $\mathrm{APN\textit{-}def}(F_{0, \alpha})=\mathrm{APN\textit{-}def}(F_{0,1}) $ for any $\alpha \in \F_{q}^*$.
\end{remark}

The next result is obtained by \"Ozbudak. 

\begin{lemma} \label{ferruh}
	 (\cite{ferruh})
Let $n$ be even, and $k$, $s$ be as defined in Theorem \ref{apndeff0a}, part II. Then, $k=s$ if $ n \equiv 2\  (\hbox{mod}  \ 4) $, and $k=s+4$	if $ n \equiv 0\  (\hbox{mod}  \ 4) $.
	
\end{lemma}

We therefore have the following simple relation between the APN-defects of the functions $F$ and $F_{0,\alpha}$.

\begin{corollary}
	\label{ozlem}
	Let $n$ be even, $\alpha \in \F_{q}^*$ be arbitrary and  $F$, 
	$F_{0, \alpha}$ be as in Theorem \ref{apndeff0a} above, with $k=|R\mathrm{\textit{-}Spec_{q}}(F_{0,\alpha})|$. Then,
		$$\mathrm{APN\textit{-}def}(F_{0,\alpha})=9(q-1)-k=\mathrm{APN\textit{-}def}(F)-k, ~\mathrm{when}~  n \equiv 2 \  (\hbox{mod}  \ 4) ,$$ and 		
		$$\mathrm{APN\textit{-}def}(F_{0, \alpha})= 9(q-1)-k-18=\mathrm{APN\textit{-}def}(F)-k-18,  ~\mathrm{when}~ n \equiv 0\  (\hbox{mod}  \ 4) .$$
	
\end{corollary}

\subsection{The partial quadruple system associated to $F_{0, \alpha}$}

In the light of Section 2.2, one can find the partial quadruple system associated to the function $F_{0, \alpha}$ easily.

\begin{theorem}
	\label{pqsf0a}
	Let $\alpha \in \F_{2^n}^*$ be arbitrary,  $F(x)= x^{{2^n}-2}$, and the permutation polynomial 	
	$F_{0, \alpha}$ be as defined in (\ref{foalfa}) above. Put $\omega\in \F_4 \setminus \F_2$, $\nu  \in R_{Q(1,1,\omega)}$, $\eta \in R_{Q(1,1,\omega^2)}$, and $\rho_a \in R_{Q(1,1,\frac{\alpha}{a + \alpha})} $, $\mu_a  \in R_{Q(1,1,\frac{\alpha}{a})}$ for a given $a \in \F_{2^n}^*$. \\

	The partial quadruple system associated to the function $F_{0, \alpha}$ is $(\F_{2^n},VF_{F_{0, \alpha}})$, where the sets $VF_{F_{0, \alpha}}$, depending on the parity of $n$, are given as follows. 
	
	\begin{itemize}
		\item[I.] Suppose $n$ is odd. Then, \\
		
		$VF_{F_{0, \alpha}}= \{\{0, a,a\rho_a, a\rho_a+a \}: a\in \F_{2^n}^*,\T(\frac{\alpha}{a + \alpha})=0\} \  \cup \   \{\{\alpha, \alpha+a,a\mu_a, a\mu_a+a \}: a\in \F_{2^n}^*,\T(\frac{\alpha}{a  })=1     \}$. \\

		\item[II.] Suppose $n$ is even. Then, \\

	$VF_{F_{0, \alpha}}=\{ \{0, \alpha, \alpha\omega,  \alpha\omega^2 \}\} \  \cup \  \{\{0, a,a\rho_a, a\rho_a+a \}: a\in \F_{2^n}\setminus \{0, \alpha, \alpha\omega,  \alpha\omega^2 \}, \\ \T(\frac{\alpha}{a + \alpha})=0    \} \ \cup \ \{\{\alpha, \alpha+a,a\mu_a, a\mu_a+a \}: a\in \F_{2^n}\setminus \{0, \alpha, \alpha\omega,  \alpha\omega^2 \},\T(\frac{\alpha}{a})=0    \}  $ when  $n \equiv 2 \ (\hbox{mod} \ 4)$, and\\

				$VF_{F_{0, \alpha}}=  \{\{0, \alpha, \alpha\omega,  \alpha\omega^2 \}, \{ 0, \alpha,  \alpha\omega^2\nu,  \alpha\omega^2(\nu+1)  \}, \{ \alpha\omega
				, \alpha\omega^2, \alpha\omega^2\nu,\\ \alpha\omega^2(\nu+1) \}, \{ 0, \alpha,  \alpha\omega^2\eta,  \alpha\omega^2(\eta+1)  \}, \{ \alpha\omega
				, \alpha\omega^2, \alpha\omega^2\eta,\\ \alpha\omega^2(\eta+1) \}  \} \ \cup \ \{\{0, a,a\rho_a, a\rho_a+a \}: a\in \F_{2^n}\setminus \{0, \alpha, \alpha\omega,  \alpha\omega^2 \}, \\ \T(\frac{\alpha}{a + \alpha})=0    \} \cup  \{\{\alpha, \alpha+a,a\mu_a, a\mu_a+a \}: a\in \F_{2^n}\setminus \{0, \alpha, \alpha\omega,  \alpha\omega^2 \},\T(\frac{\alpha}{a})=0\}$
						when  $n \equiv 0 \ (\hbox{mod} \ 4)$.\\
			
		\item[III.]	Moreover, recalling that $q=2^n$, $|VF_{F_{0, \alpha}}|$ is given as follows.\\
		
		\item[III.i] When $n$ is odd,	$$|VF_{F_{0, \alpha}}|=\frac{1}{3}(q+\ell-k-1),$$  where $k=|R\mathrm{\textit{-}Spec_{q}}(F_{0,\alpha})|$, and  $\ell=|\{a \in \F_{2^n}\setminus\{0,\alpha\}: \T(\frac{\alpha}{a + \alpha})=0, \  \mathrm{and} \  \T(\frac{\alpha}{a})=1  \}|$.\\
		
		\item[III.ii] When  $n \equiv 2 \ (\hbox{mod} \ 4)$,
		
		$$|VF_{F_{0, \alpha}}|=\frac{1}{3}(q+s-k-1),$$
		 where $k=|R\mathrm{\textit{-}Spec_{q}}(F_{0,\alpha})|$, and 
		 $s=|\{a \in \F_{2^n}\setminus \{0,\alpha, \alpha\omega, \alpha \omega^2\}: \T(\frac{\alpha}{a + \alpha})=\T(\frac{\alpha}{a})=0  \}|$.\\
		
		\item[III.iii] When  $n \equiv 0 \ (\hbox{mod} \ 4)$,
		
		$$|VF_{F_{0, \alpha}}|=\frac{1}{3}(q+s-k+3),$$
		 where $k=|R\mathrm{\textit{-}Spec_{q}}(F_{0,\alpha})|$, and 
		 $s=|\{a \in \F_{2^n}\setminus \{0,\alpha, \alpha\omega, \alpha \omega^2\}: \T(\frac{\alpha}{a + \alpha})=\T(\frac{\alpha}{a})=0  \}|$. \\
		 	\end{itemize}	
		 	
		 \end{theorem}
		 
	\begin{proof}
Parts \textit{(I)} and \textit{(II)} follow from Theorem \ref{f0a} and the proof of Lemma \ref{pqs}. Indeed, for those values of $a$ and $x$ with $\nabla(a,x)=4$, the set $ D_{F_{0,x}(a,x)}$, consisting of 4 elements form the vanishing flat containing  $x$. When $\nabla(a,x)=6$, the set $ D_{F_{0,x}(a,x)}$ gives rise to 3 distinct vanishing flats, as explained in the proof of Lemma \ref{pqs}.

In order to prove part \textit{(III)}, we recall that each vanishing flat occurs exactly three times, see the proof of Theorem (II.3) in \cite{vanishing}. In case of odd $n$ and even $n$ satisfying $n \equiv 2 \ (\hbox{mod} \ 4)$, it is sufficient to count the number of elements in  $\{a \in \F_{2^n}^*:\nabla(a,x)=4 \ \hbox{for some} \  x \in \F_{2^n} \}$. Therefore 
we need to consider the cases \textit{(II)} and \textit{(IV)} in Theorem 
\ref{f0a}. Using the notation of Theorem \ref{apndeff0a}, we have $2\ell+(q-k-\ell-1)$, and $2s+(q-k-s-1)$ such instances for odd $n$
and even $n$ satisfying $n \equiv 2 \ (\hbox{mod} \ 4)$, respectively. Dividing by three, we obtain the number of distinct vanishing flats.

When $n \equiv 0 \ (\hbox{mod} \ 4)$, we also need to consider the cases 
$a \in \{\alpha\omega,  \alpha\omega^2\}$, where $\nabla(a,0)=\nabla(a,\alpha)=6. $ We therefore have $2s+(q-k-s-3)+2\cdot 3$
vanishing flats, one third of which are distinct. 
	$\hfill\square$	
	\end{proof}

\begin{remark}
	\label{newpqs}
We would like to emphasize that the use of our method, in particular Lemma \ref{pqs} and Theorem \ref{f0a} above, enables us to produce (as far as we are aware) the very first example of the partial quadruple system associated to a function, which is not a power function or a DO polynomial.    
\end{remark}

\textit{Example \ref{vffoa} re-visited.} The function $G$ in Example \ref{vffoa} is 	$F_{0, \alpha}$, and hence the argument concerning the vanishing flats with respect to $G$ exemplifies Theorem \ref{pqsf0a}, in the case $n=16$, i.e., 
$n \equiv 0 \ (\hbox{mod} \ 4)$. \\

Combining Theorem \ref{pqsf0a}, Lemma \ref{invlemma}, and Lemma \ref{ferruh}, we get the following result on the number of vanishing flats $|VF_{F_{0, \alpha}}|$ when $n$ is even.	\\
\begin{corollary}
	\label{numvf}
	Let $n$ be even, $q=2^n$. With the terminology used in Theorem \ref{pqsf0a} we have,  
	$$|VF_{F_{0, \alpha}}|=|VF_{F}|=\frac{1}{3}(q-1).$$ 
\end{corollary}

\section{Overview}
In this section, we emphasize the relations between the concepts that we introduced in Sections 2 and 3, notably, the $(p_a)$ property, $x_0$-partial APN-ness and the vanishing flats, recall Definitions $1, 2, 3, 1^*, 2^*$.  We also explain why they are relevant to our work and how APN-defect extends the information gained through them. \\

We start with the following observation about power functions. 

\begin{lemma}
\label{p?}
Let $G(x)=x^d$ be a power function of $\F_{2^n}$. The following are equivalent. 
\begin{itemize}
\item[I.] G is APN.
\item[II.] $G$ satisfies the property $(p_1)$.
\item[III.] $G$ is $1$-pAPN.
\item[IV.] $|S_1|=2^n$, where $S_1$ is as defined in (\ref{sa}).
	
\end{itemize}

\end{lemma}

\begin{proof}
	The proof is trivial since it is known that \textit{(I)} and \textit{(III)} are equivalent, see Theorem 4.4 in \cite{papn} and obviously, \textit{(I)} is equivalent to \textit{(II)} and \textit{(IV)}.  However, 
we prove that \textit{(II)} and \textit{(III)} are equivalent since the direct proof is very simple.
 
To show that  \textit{(II)} implies \textit{(III)}, we suppose the contrary, i.e., that $G$ is not $1$-pAPN. In this case there exist $a \in \F_{2^n}^*$ and $x \in \F_{2^n}, x \neq 1, 1+a,$ such that $$1+(1+a)^d=x^d+(x+a)^d,$$ and one obtains 
$$y_1^d+(y_1+1)^d=y_2^d+(y_2+1)^d,$$ where $y_1=1/a, y_2=x/a$. 
 Hence $G$ does not satisfy the property $(p_1)$. The converse follows similarly.
 	$\hfill\square$		
\end{proof}

	Theorem \ref{pqsf0a}, especially in part \textit{(III)}, demonstrates the relation between APN-defect and number of vanishing flats. The next result makes this relation explicit. 
	
	\begin{theorem}
		\label{defvsvf}
	Given $G:\F_q \rightarrow \F_q$, $q=2^n$. Then,
	
%	\begin{equation}
	%\label{eqdefvsvf}
	$$ \mathrm{APN\textit{-}def}(G)=q-12|VF_G|+ \sum_{a\in \F_{q}^*}(3w_{S_a^c}-\chi_a)
	 -1.$$	 
	% \mathrm{APN\textit{-}def}(G)=-12|VF_G|+ \sum_{a\in \F_{q}^*}(3w_{S_a^c}-\chi_a)+q-1. 
	%	\end{equation} 

	\end{theorem}
	
	\begin{proof} We recall Equations (\ref{dela}) and (\ref{delwa});

$$|S_a|=2 \sum_{b\in \F_{q}, \delta_G(a,b)>0}{2 \choose \delta_G(a,b)},$$

$$w_{S_a^c}=\sum_{b\in \F_{q},\delta_G(a,b) > 2 }{\left(\frac{\delta_G(a,b)}{2}\right)^2}.$$

On the other hand, the number of vanishing flats $|VF_G|$ satisfies, see Theorem (II.3) in \cite{vanishing},

$$|VF_G|=\frac{1}{3} \sum_{a\in \F_{q}^*, b\in \F_{q}}{\delta_G(a,b)/2 \choose 2}.$$

Note that,

$$\sum_{b\in \F_{2^n}}{\delta_G(a,b)/2 \choose 2}=\frac{1}{2}w_{S_a^c}-\frac{1}{4}\sum_{b\in \F_{q}, \delta_G(a,b) > 2} \delta_G(a,b), $$

and

$$\sum_{b\in \F_{q}, \delta_G(a,b) > 2} \delta_G(a,b)+ \sum_{b\in \F_{q}, \delta_G(a,b) = 2} \delta_G(a,b)=q.$$
	
Therefore,

$$ \sum_{b\in \F_{q}, \delta_G(a,b) > 2} \delta_G(a,b)=q-|S_a|, $$		
and 

$$ \sum_{a\in \F_{q}^*, b\in \F_{q}}{\delta_G(a,b)/2 \choose 2}=\frac{1}{4} \sum_{a\in \F_{q}^*}(|S_a|+2 w_{S_a^c}-q),$$
implying

$$12|VF_G|=\sum_{a\in \F_{q}^*}(|S_a|+2w_{S_a^c})-q(q-1),$$
and the assertion follows. 
$\hfill\square$		
	
	\end{proof}

	\begin{remark}
			\label{rem15}
We have shown in Theorem \ref{power} above that the APN-defect of the 	inverse function $F(x)=x^{q-2}$ is $9(q-1)$ when $n$ is even, see (\ref{defect-mono}). By using the relation given in Theorem \ref{defvsvf}, one can retrieve the number of vanishing flats as
$|VF_F|=\frac{1}{3} (q-1) $.

	\end{remark}
	
	We now focus on a concept, introduced in \cite{localapn}, which is concerned with power functions: A function $G:\F_{2^n} \rightarrow \F_{2^n},  G(x)=x^d$ is called \textit{locally-APN} if $\delta_G(1,b) \leq 2$ for all $b \in \F_{2^n}\setminus \F_2.$ In our terminology, $G$ is locally-APN if $\nabla_{G}(1, x)=2$ for all $x \in \{y: D_1G(y)\neq 0, 1  \}$, i.e., the entries in the first row $\Delta_1(G)$ of the difference square repeat exactly twice, except for the entries with values $0$ and $1$.
	
	\begin{example}
	Consider $F(x)=x^{2^n-2},$ the inverse function. It is locally-APN since the value $1$ repeats 4 times in the first row $\Delta_1(F)$ of the difference square, and the remaining values occur exactly twice, see Remark \ref{reminv}. 
	\end{example}
	
	The following results indicate a connection between local APN-ness and partial APN-ness. The proofs are simple, where we use the point of view of difference squares.  
	
	\begin{lemma}
		\label{locapn1}
	Let $G(x)=x^d$ be a permutation with $\delta_G \geq4.$ If $G$ is locally-APN, then it is not $0$-pAPN.
	\end{lemma} 
	
	\begin{proof}
It is shown in \cite{localapn}, Lemma 1 that $\delta(1,0)=0$ when $G$ is a permutation. Therefore, $\delta(1,1)=|\{y\in \F_{2^n}: D_1G(y)=1\}|\geq 4$ since $G$ is locally-APN and $\delta_G \geq4.$ Note that $D_1G(0)=1$, implying that 
$\nabla(1,0)\geq 4$. Hence $G$  is not $0$-pAPN by Definition $2^*$.
	$\hfill\square$			\end{proof}

The next observation follows immediately by the argument used in the above proof.	
	\begin{corollary}
		\label{localapn}
		Suppose $G(x)=x^d$ is a locally-APN permutation. If it is $0$-pAPN,
		then it is APN.
	\end{corollary}

	Corollary \ref{localapn} gives a criterion for APN-ness. The $(p_a)$ property and $x_0$-partial APN-ness can also be used for the same purpose. 
	 This aspect of the $(p_a)$ property for any function over $\F_{2^n}$ is explained in the paragraph following Definition 1 above. We refer the reader to \cite[Proposition 4.1]{papn}  for such a criterion concerning partial APN-ness, where it is shown that a power function $G$ is APN if and only if it is $0$-pAPN and $x_1$-pAPN for some $x_1 \in \F_{2^n}^*$, see also Lemma \ref{p?} and Corollary \ref{localapn} above. We note that the concepts of APN-defect and vanishing flats cannot be used for checking APN-ness.\\

	On the other hand, when APN-defect is used as a tool for measuring the distance of a given function $G$ to the set of APN functions,
	it has some favourable properties that we briefly explain below.\\
	
	We mentioned earlier, see Remark \ref{remark5}, %and Remark \ref{rem14}, 
	that the vanishing flats are concerned (exclusively) with $a \in \F_{2^n}^*, x \in \F_{2^n}$ with $\nabla(a,x) \geq 4.$
The APN-defect on the other hand, takes into account $a \in \F_{2^n}^*, x \in \F_{2^n}$ satisfying $\nabla(a,x)=2.$ This aspect of the APN-defect proves to be advantageous. For instance, Remark \ref{remark5} points to a power function, differential spectrum of which depends on the divisibility of $n$ by 3 while the number of vanishing flats remains the same for any odd $n\geq 7$. However, the APN-defect also varies depending on the divisibility condition.\\

%Arguably, a more important feature of the APN-defect is indicated in 
Another interesting instance is indicated in Example \ref{vffoa}. Indeed, it is shown that  $|VF_F|=|VF_{F_{0,\alpha}}|=5$, for $q=16$ (see also Corollary \ref{numvf}), although $\delta_F=4$ and $\delta_{F_{0,\alpha}}=6$, and hence these two functions are CCZ-inequivalent.  We have $\mathrm{APN\textit{-}def}(F)=9 \cdot 15=135$, and 
$\mathrm{APN\textit{-}def}(F_{0,\alpha})=9(q-4)+5=113$ ($k=4, s=0$ in this case, see Example \ref{vffoa}, Theorem \ref{apndeff0a} or Corollary \ref{ozlem}). \\

We also note that the differential spectra of $F$ and $F_{0,\alpha}$ are the same for $n=6$ and $n=10.$ Interestingly, Corollary \ref{spectf0alpha} implies that $|R\mathrm{\textit{-}Spec_{2^6}}(F_{0,\alpha})|=12$ for any $\alpha \in \F_{2^6}^*$, while $|R\mathrm{\textit{-}Spec_{2^6}}(F)|=0$. Moreover, they can be shown to be CCZ-inequivalent and again, 
\begin{equation}
\label{apndiff}
 \mathrm{APN\textit{-}def}(F) \neq \mathrm{APN\textit{-}def}(F_{0,\alpha})
 \end{equation} 
  over $\F_{2^6}$ and $\F_{2^{10}}$. Indeed,  $|R\mathrm{\textit{-}Spec_{2^6}}(F)|, |R\mathrm{\textit{-}Spec_{2^{10}}}(F)|>0$, and Corollary \ref{ozlem} implies (\ref{apndiff}). The capability of the APN-defect in distinguishing CCZ-inequivalent functions will be presented in detail in the second part of this work.\\
	
%The information gained by the APN-defect of a non-APN function is more comprehensive than that obtained by other similar measures. For instance,  %$(p_a)$ property also, in the sense that 
We recall that the cardinality of the row spectrum is one of the inputs for calculating the APN-defect. Note that the APN-defect carries information about the number of 2-to-1 derivatives in the rows $\Delta_a(F)$, where the $(p_a)$ property does not hold. Moreover, it indicates a relation between the number of 2-to-1 derivatives and $k$-to-1 derivatives for $k \geq 4$. In this sense, it adds to the knowledge provided by the $(p_a)$ property.
\\

%A similar comparison holds for 
The information gained by the APN-defect of a non-APN function can be more comprehensive than that obtained by
$x_0$-partial APN-ness in some cases. Firstly, observe %give a simple proof of the fact 
that $C\mathrm{\textit{-}Spec_{q}}(F)=\emptyset$ for $F(x)=x^{2^n-2}$, when $n$ is even. For a simple proof of this fact, recall that $D_aF(0)=D_aF(a\omega)$ for $\omega \in \F_4\setminus\F_2$, and any $a \in \F_q^*. $ Hence $F$ is not $0$-partially APN. Now assuming $x_0 \neq 0$ is arbitrary, we put $x_0\omega^{-1}=a_0.$ Then one has $D_{a_0}F(0)=D_{a_0}F({a_0}\omega)=D_{a_0}F(x_0)$, showing that $F$ is not $x_0$-partially APN. Therefore, $F$ behaves in the worst possible way with respect to $x_0$-partial APN-ness.
 On the other hand, as mentioned earlier, the differential behaviour of $F$ is the best possible among non-APN power functions according to other measures including the number of vanishing flats, the number of 2-to-1 derivatives, and the APN-defect. \\ 
 
 Similarly, $C\mathrm{\textit{-}Spec_{q}}(F_{0,\alpha})=\emptyset$ when $n$ is odd, while Theorem 5, parts (I)-(III) demonstrate that $F_{0,\alpha}$ has rather nice differential properties (for instance, in terms of the number of 2-to-1 derivatives) and this feature is reflected by the APN-defect, see Theorem 6 part I, which shows that ${\cal{D}}(F_{0,\alpha})$ is positive. \\

 In the light of the above comments, APN-defect appears to be a rather effective and a favourable tool in categorizing non-APN functions according to their \textquotedblleft distances\textquotedblright$~$to the set of APN-functions. \\

\section{Conclusion}

We study non-APN functions with the aim of identifying those which behave favourably in terms of  their differential properties. We introduce a new measure for this purpose, the APN-defect, which assesses functions $G:\F_{2^n} \rightarrow \F_{2^n}$ based on the number of 2-to-1 derivatives $D_a G(x)$ when $a\in \F_{2^n}^*$ varies, relative (weighted) frequency of 2-to-1 and $k$-to-1 derivatives for $k\geq4$, and the cardinality of the row spectrum of $G$. As exemplified in Section 5, consideration of all these parameters together has the advantage that the APN-defect yields potentially more information on the differential behaviour of $G$ than that gained by similar measures. Moreover, the use of difference squares in analyzing these parameters enables us to simplify and/or clarify arguments that lead to proofs of new (and old) results. \\

We present our work in two parts with the following aims. On the one hand, we wish to accommodate detailed descriptions and examples in order to enhance the readability. On the other hand, we would like to to avoid the manuscript to be too lengthy. \\

In the forthcoming next part we discuss properties of modifications of the inverse function with higher Carlitz ranks, and address questions concerning CCZ-equivalence. We also study the behaviour of a class of functions over infinitely many extensions of $\F_{2^n}$ in connection with their APN-defects. 

\section{Acknowledgement}

N. Anbar and T. Kalayc\i $\;$   
are supported by T\"UB\.ITAK Project under Grant 120F309.

\end{document}